\documentclass[showpacs,onecolumn,superscriptaddress,floatfix,nofootinbib,11pt,a4paper,accepted=2023-08-22]{quantumarticle}
\pdfoutput=1 
\usepackage[margin=2.8cm]{geometry}
\usepackage{amsfonts}
\usepackage{graphicx}
\usepackage{caption}
\usepackage{thmtools,thm-restate}
\usepackage{amssymb}
\usepackage{amsmath}
\usepackage{amsthm}
\usepackage{mathrsfs}
\usepackage{bm}% bold math
\usepackage{gensymb}
\usepackage{algorithm}
\usepackage{algpseudocode}
\usepackage{booktabs}
\usepackage{setspace}
\usepackage[braket,qm]{qcircuit}
\usepackage{mathtools}
\usepackage{stmaryrd} % \xmapsto: \mapsto with text above
\usepackage[pdftex,colorlinks=true,linkcolor=black,citecolor=blue,urlcolor=black]{hyperref}
\usepackage[table,xcdraw]{xcolor}
\usepackage{adjustbox}
\usepackage{makecell}
\usepackage{caption}
\usepackage{subcaption}
\usepackage[normalem]{ulem}
\usepackage{float}
\newfloat{algorithm}{t}{lop}
\graphicspath{ {./figures/} }
\usepackage{multirow}
\usepackage[multiple]{footmisc}

\usepackage{mathtools}

\usepackage{authblk}

\usepackage{amsthm}

\newtheorem{lemma}{Lemma}
\newtheorem{corollary}{Corollary}

\newcommand{\mO}{\mathcal{O}}

\newcommand{\later}[1]{}

\newcommand{\qsb}{\textbf{QSearch}_{\infty}}
\newcommand{\qmi}{\textbf{QMax}_{\infty}}

\newcommand{\R}{\mathbb{R}}
\newcommand{\E}{\mathbb{E}}
\newcommand{\bbN}{\mathbb{N}}

\newcommand{\f}{\text{fail}}
\newcommand{\s}{\text{success}}

\newcommand{\tout}{\text{timeout}}
\newcommand{\Qm}{Q_{\text{max}}}

\DeclareMathOperator*{\argmax}{arg\,max}

\newcommand\numberthis{\addtocounter{equation}{1}\tag{\theequation}}

\DeclarePairedDelimiter\ceil{\lceil}{\rceil}

\DeclarePairedDelimiter\round{\lceil}{\rfloor}

\newcommand{\ccrr}{\cellcolor[HTML]{FF8F8C}}

\newcommand{\ccgg}{\cellcolor[HTML]{87DC86}}

\newcommand{\ZQ}{$\textbf{QSearch}_{\text{Zalka}}$}

\newcommand{\maxsat}{{\sc max-}$k${\sc -sat} }
\begin{document}

\title{Quantifying Grover speed-ups beyond asymptotic analysis}

\author[1,2]{Chris Cade}
\author[3]{Marten Folkertsma}
\author[1,2]{Ido Niesen}
\author[3]{Jordi Weggemans}
\affil[1]{QuSoft \& University of Amsterdam (UvA), Amsterdam, the Netherlands}
\affil[2]{Fermioniq, Amsterdam, the Netherlands}
\affil[3]{QuSoft \& CWI, Amsterdam, the Netherlands}

\date{September 19, 2023}

\maketitle

\begin{abstract}
\noindent
Run-times of quantum algorithms are often studied via an asymptotic, worst-case analysis. Whilst useful, such a comparison can often fall short: it is not uncommon for algorithms with a large worst-case run-time to end up performing well on instances of practical interest. To remedy this it is necessary to resort to run-time analyses of a more empirical nature, which for sufficiently small input sizes can be performed on a quantum device or a simulation thereof. For larger input sizes, alternative approaches are required.

In this paper we consider an approach that combines classical emulation with detailed complexity bounds that include all constants. We simulate quantum algorithms by running classical versions of the sub-routines, whilst simultaneously collecting information about what the run-time of the quantum routine would have been if it were run instead. To do this accurately and efficiently for very large input sizes, we describe an estimation procedure and prove that it obtains upper bounds on the true expected complexity of the quantum algorithms. 

We apply our method to some simple quantum speedups of classical heuristic algorithms for solving the well-studied MAX-$k$-SAT optimization problem. This requires rigorous bounds (including all constants) on the expected- and worst-case complexities of two important quantum sub-routines: Grover search with an unknown number of marked items, and quantum maximum-finding. These improve upon existing results and might be of broader interest. 

Amongst other results, we found that the classical heuristic algorithms we studied did not offer significant quantum speedups despite the existence of a theoretical per-step speedup. This suggests that an empirical analysis such as the one we implement in this paper already yields insights beyond those that can be seen by an asymptotic analysis alone.

\end{abstract}

\setcounter{tocdepth}{2}

\newpage
\tableofcontents
\newpage
% TO DO
% - add edge-community pair Louvain : Ido -- actually, let's keep it at one and just add a comment in the numerics we also looked at searching over vertex-community pairs, and its the same
% - all quantum algos have theorem
% - Nested Grover: add a comment somewhere about this issue of delayed measurements and the ridiculous space requirements it entails
% - description of simulation of algorithms : Jordi & Marten
% - switch classical to quantum in FindFirst : Marten

% - Future directions: boosting vs span programs, be aware of extra overhead due to mapping to span program <- what does this mean?
% - run-time vs queries: make consistent

\section{Introduction}

\noindent There is growing motivation to design and evaluate quantum algorithms for commercial applications in order to assess the potential impact of quantum computers. To determine whether, or when, a quantum algorithm should be used for a task will involve comparing the candidate quantum algorithm, or set of algorithms, to an existing state-of-the-art classical one. A common approach to benchmark and compare algorithms is to consider their performances on worst-case instances, by providing upper bounds to their runtimes: bounds that hold for every possible instance of the problem the algorithm is designed to solve. In this context, a quantum speedup of a classical algorithm refers to the use of quantum algorithmic techniques that give an improvement over the worst-case runtime of the classical algorithm in question. In some cases, \emph{expected} run-times are considered, where the expectation is taken over the internal (classical or quantum) randomness of the algorithm, and then upper bounds on this expectation are compared. In even rarer cases, it is possible to rigorously analyse the \emph{average}-case complexity of the algorithms, where now the average is taken over the set of inputs~\cite{ben1992theory}. 

However, such worst-case (and to a lesser extent, average-case) upper bounds can often be misleading: it is not uncommon for algorithms with a large worst-case run-time to perform very well in practice~\cite{spielman2009smoothed,wright2005interior}.\footnote{Interestingly, this observation can actually be made rigorous for some algorithms in certain situations~\cite{karp1985probabilistic}.} For instance, this is especially true of heuristic algorithms, which are commonly used to solve real-world problems and are often fine-tuned to perform well on instances of interest, rather than on an artificial instance designed to be as difficult as possible for the algorithm but that will likely not appear in a natural setting. As such, much of quantum algorithms research focuses either on exponential quantum speedups (e.g. Shor's algorithm), in which case the speedup obtained by the algorithm is unambiguous; or when only a modest quantum speedup is available, on situations where the run-time of both the quantum and classical algorithms can be determined in a reasonably tight way: the square-root speedup obtained by Grover's algorithm for unstructured search is a simple example of this. For more complicated (quantum and classical) algorithms, however, it might be that the run-time suggested by an asymptotic analysis fails to capture the true complexity of the algorithm on inputs that will be encountered in practice, which can make it difficult to determine the usefulness of a candidate quantum algorithm over a classical one. A similar observation was made in~\cite{schuld2022quantum}, where the authors point out that this disconnect is one of the main reasons that it is so difficult to design quantum algorithms for machine learning and assess their performance relative to their classical counterparts. 

In this paper we continue along a line of work that moves beyond performing purely asymptotic analyses of quantum algorithms towards ones of a more empirical nature. In the time before large fault-tolerant quantum computers become readily available, we suggest that for the majority of quantum algorithms, an intermediate form of classical simulation + run-time estimation is possible, and that it can allow for meaningful and informative comparisons to be made. Our particular approach combines tight asymptotic analysis with classical simulation, in an attempt to carefully estimate the run-time of quantum algorithms in lieu of actually being able to run them on a quantum device. Importantly, our methodology is sensitive to the input given to the quantum algorithm. 

To verify the utility of our approach, we perform such an empirical analysis for a set of reasonably simple quantum versions of a classical heuristic algorithm for a particular use-case: that of {\sc MAX-$k$-SAT}. The quantum speedups we obtain are quite typical of quantum speedups of classical optimisation algorithms: the classical routine repeats a number of steps, the $k$th taking some time $t_k$ -- which will depend on what happened in previous steps -- until convergence; the quantum algorithm does the same, except with each step taking time now $\approx\sqrt{t_k}$. To assess the usefulness of such a quantum algorithm, we must ask: to what extent does this square-root-like speedup manifest in the algorithm when it is actually run to convergence? Moreover, it is likely that the behaviour of the algorithm will differ substantially on different inputs, and hence a further question we should ask is: how much of a speedup does the quantum algorithm obtain on a \textit{representative or real-world} input? We seek to answer such questions with our empirical approach.

% and encounter along the way some important considerations and subtleties that must be dealt with appropriately in order to make this approach useful.

% We would like to emphasize that the approach that we take will be most useful, and informative, for studying heuristic quantum algorithms and comparing them to their classical counterparts. 

\subsection{Summary of results}
Our main results and contributions are:
\begin{itemize}
    \item Improved analyses of upper bounds on the expected and worst-case complexities of Grover search when the number of marked items is unknown including log and constant factors, improving upon analyses performed in earlier works (e.g.~\cite{zalka1999grover,boyer1998tight}). We also consider how to optimize the number of classical samples drawn before Grover iterations are used. Sections~\ref{sec:qsearch_expected_case},~\ref{sec:qsearch_worst_case}.
    \item Upper bounds on the expected complexity of a quantum maximum finding algorithm, improving upon those in previous works (e.g.~\cite{durr1996quantum,ahuja1999quantum}). Section~\ref{sec:maxfinding}.
    \item An estimation procedure that allows us to use the above bounds to obtain estimates of the expected run-times of repeated calls to a Grover search sub-routine when the number of marked items cannot be computed exactly (in our classical simulations), something that is useful (and indeed necessary) for bench-marking quantum algorithms on very large inputs. The outputs of the procedure come with theoretical guarantees. Section~\ref{sec:simulating_quantum_algorithms}.
    \item A general approach combining the above that allows for rigorous (and efficient) classical estimation of the run-times of quantum algorithms that make repeated calls to Grover sub-routines. This is achieved via classical emulation of the underlying quantum algorithms. 
    % The code for implementing such simulations is available at {\bf [Github]}.
    \item Two simple quantum heuristic algorithms for {\sc MAX-$k$-SAT}, which are basic quantizations of classical `hill climber' algorithms. Section~\ref{sec:quantum_maxsat_algo}. 
    \item We find that the quantum hill climbers obtain favourable scaling compared to their classical counterparts, but that only one of them (the `simple' hill climber) obtained an absolute speedup for the problem sizes we considered. We observe that some, but not all, of the per-step speedup indicated by an asymptotic analysis manifests in the final behaviours of the algorithms. Section~\ref{sec:numericsss}.
    \item We verify that our estimation procedure does indeed yield accurate estimates of the expected run-times of our algorithms when compared to an exact method. Section~\ref{sec:numericsss}.
\end{itemize}

\subsection{Concurrent work}
\label{sec:concwork}
In a concurrent work~\cite{cade2022community}, we apply our methodology to help design quantum algorithms for a common task in complex network analysis. The quantitative analysis employed in this other study is notably more comprehensive than the one employed for the elementary hill-climbing algorithm examined in this paper, and thus, the two studies are complementary: the former serves as an introductory exposition to the methodology, while the latter showcases its effectiveness when utilized for developing quantum algorithms for practical problems. More specifically, the other work studies the \emph{Louvain algorithm}\footnote{The algorithm takes the name of the city in which it was developed. The original paper describing the method has been cited over 15,000 times, and the algorithm itself can be found in all popular graph/network analysis software packages.}, which forms one of the main tools for tackling a problem ubiquitous to the study of complex networks: that of \emph{community detection}. Together with its descendants, the Louvain algorithm has successfully been used to study large sparse networks with millions of vertices~\cite{blondel08,de2011generalized,que2015scalable,Traag2019}. In~\cite{cade2022community}, we introduce several quantum versions of the Louvain algorithm, analyse their asymptotic (worst-case) complexities, and investigate numerically how they perform on randomly generated networks, as well as on real-world data sets.

\subsection{A broader perspective}

We remark that the kind of analysis we perform should in principle be possible for all quantum algorithms that achieve small polynomial speedups over classical ones: we can always `simulate' the quantum algorithm by running its classical equivalent (which will be only polynomially slower!), and simultaneously estimate how long the quantum routine would have made if it were run instead. All that is required are appropriately tight bounds (including constants, etc.)~on the run-times of the quantum sub-routines used by the algorithm. With all of the above in mind, the semi-empirical approach to practical quantum algorithm design and analysis that we use fits into a larger framework that follows the following structure:
\begin{itemize}
    \item \textbf{Design a quantum algorithm} or collection of algorithms, perhaps via speedup of an existing classical algorithm.
    \item \textbf{Choose a measure of complexity} for the algorithms, ideally one that is agnostic about the capabilities of future hardware\footnote{For reasons explained in Section~\ref{sec:intro-methodology}.}. This could be for instance the number of time-steps, or the number of queries to the input or to some function.
    % \item \textbf{[Optional] Perform a worst/expected-case asymptotic analysis} in the usual way. This can be seen as a first step towards seeing if there could be a quantum speedup -- if there is no speedup to be seen at this stage then almost certainly there will not be one in practice.
    \item \textbf{`Simulate' the quantum algorithms on inputs of interest} by replacing the quantum routines with their classical counterparts, and instead collect information to estimate what the quantum complexities would have been if those sub-routines were used instead. This will require one to obtain or prove (ideally tight) bounds on the worst/expected-case complexities of the quantum subroutines used by the algorithms.
    \item \textbf{Use these empirical results} to inform the choice or design of the quantum algorithms. For instance, one might observe that a particular quantum algorithm can be made faster in practice by simplifying it and sacrificing some asymptotic speedup.
\end{itemize}

As we will see in the sections that follow, there are further considerations that must be taken into account that can prove to be tricky even for simple algorithms such as the ones we consider, suggesting that such an analysis is unlikely to be entirely straightforward in general. Nevertheless, a very fruitful next step could be to build such `pseudo-simulation' of quantum algorithms into any of the number of quantum programming languages now available~\cite{bichsel2020silq,hancock2019cirq,steiger2018projectq,svore2018q,wille2019ibm}, which might allow for these sorts of empirical analyses to be performed more quickly and painlessly, and hence facilitate faster quantum algorithm development and prototyping.

\subsection{Methodology}
\label{sec:intro-methodology}

Here we describe our methodology for comparing the run-times of quantum and classical algorithms. The most obvious way to do this is to run both algorithms on their respective devices and measure the time they take to run to completion. Unfortunately, quantum hardware is currently not sufficiently developed to be able to run any of the algorithms we describe, and therefore such a comparison is not possible at this point, and will likely not be for the foreseeable future. An alternative approach is to simulate the quantum algorithm at the qubit level on classical hardware, count the number of quantum gates applied and then compare this to the required number of classical gates. This is often the approach currently taken with heuristic quantum algorithms such as VQE~\cite{wecker2015progress,dallaire2019low,cade2020strategies} and QAOA~\cite{zhou2020quantum,weggemans2021solving}. However, such simulations are almost always going to be restricted to a few qubits, which means that a comparison between the classical and quantum algorithms can only be made for very small input sizes. Since we are interested in investigating how well the classical and quantum versions of our algorithms compare on actual datasets, which will generally be very large, this method of comparison is insufficient.

Moreover, since quantum technologies are still in their infancy, it is not unlikely that they will improve significantly over the coming years. With this prospect in mind, a comparison that depends heavily on the properties of current-day (or even current-day predictions of) quantum hardware might become obsolete in the near future. For this reason, we aim to make our comparisons \emph{architecture independent}: this will in particular mean not explicitly counting the number of gates needed to implement the algorithms, or taking into account the overheads from error correction. Hence, our comparisons will be of a more qualitative nature than a quantitative one: we are interested, in principle, in how much of the speedup suggested by an asymptotic analysis manifests in the final behaviour of the algorithm -- if no speedup appears at this stage, then it certainly won't appear \emph{after} taking into account the aforementioned overheads.

\

\noindent In lieu of estimating actual running times for our algorithms, we fix a suitable notion of complexity and use this to directly compare the classical and quantum algorithms. In particular we opt to count the number of calls made to a particular function (in fact, the very function we are trying to maximise). This essentially equates to measuring query complexity, where we count queries to a function rather than to, say, the input. Counting the number of function calls of course does not capture every costly component of the algorithms that we consider: there are parts that add to the run-time but do not require function calls. However, as is common in the study of query complexity, we choose a suitable measure of complexity so that these parts are those for which we do not obtain any quantum speedup, and hence cost the same quantumly as they do classically -- things such as updating what is stored in memory after completing a step of the algorithm. From the perspective of quantum speedups, comparing the number of function calls made by the quantum and classical algorithms can therefore serve as a proxy for how much of a speedup we can expect to gain on the part of the classical algorithm that admits a speedup. Finally, we note that choosing this complexity measure preserves the architecture independence that we strive for in our analysis, by, for example, ignoring precisely how long a memory update takes, or how many items can be retrieved from memory in a single computational `step'.

\ 

\noindent For simplicity, we focus our attention on quantum algorithms that are composed of a number of steps, each of which consists of some classical computation as well as one or more calls to a Grover search on a list containing an unknown number of marked items, and/or to a quantum maximum-finding sub-routine\footnote{I.e. the algorithm looks for either \textit{some} solution at each step (Grover search), or the \textit{best} one (quantum maximum-finding).}. The quantum algorithms for MAX-SAT discussed in Section~\ref{sec:maxsat} are examples of such an algorithm. We consider situations in which the list itself and the number of marked items in it will differ for each step, and in particular will depend on the outcome of the calls that came before it, making the behaviour of the algorithm sensitive to the input itself, as well as the (possibly random) outcomes of the processing during each step. Precisely, we consider quantum algorithms with the structure of Algorithm~\ref{alg:general}. \later{Can we make an argument that this structure captures a number of existing quantum algorithms? E.g. quantum algorithms for dynamic programming etc. Some examples}
\begin{algorithm}[!htb]
  \caption{Generic quantum algorithm structure}
  \label{alg:general}
\begin{algorithmic}[1]
    \State \textbf{Input} $X$, \textbf{Memory} $M$
    \For{$k = 1, \dots, T$}
        \State Do some classical processing on $X$ and $M$, resulting in some list $L_k$ containing $t_k$ marked items.
        \State Perform either one or more (perhaps nested) Grover searches with an unknown number of marked items on $L_k$, or run quantum maximum-finding on the list $L_k$, to obtain some item $x_k$.
        \State Do some more classical processing given $x_k$, update $M$.
    \EndFor
\end{algorithmic}
\end{algorithm}

In order to estimate the run-time of such an algorithm given some particular input we would, following our approach, execute all steps except step 4 of Algorithm~\ref{alg:general} as they would normally be executed classically, but then replace the step 4 with its classical alternative, and instead estimate how long it \textit{would have} taken if the quantum routine were called. In this way, we can estimate the run-times for different inputs of any quantum algorithm that follows this basic structure.

\later{Merge paragraph below with paragraph above Lemma 1?}
As we will see in Section~\ref{sec:methodology}, even to estimate the complexities of algorithms that make use only of Grover search and quantum maximum finding already requires a somewhat substantial effort. There we prove rigorous upper bounds (including constants) on the expected and worst-case query complexities of Grover search with an unknown number of marked items -- something that, to our knowledge, has not been done elsewhere.\footnote{Not for the expected complexity, and for the worst-case complexity not in a way that is sufficient to estimate the number of queries made by the quantum algorithm, including all constant and logarithmic overheads, that holds for all input sizes}. Using these bounds, we obtain bounds on the expected complexity of quantum maximum finding, improving upon previous results from the literature. 

Our approach can of course be extended to more complicated algorithms that make use of different quantum sub-routines by proving analogous bounds for those sub-routines. Here, however, we keep our focus narrowly on quantum algorithms of the simple structure described above, so that we can apply and demonstrate the usefulness of our approach.

Finally, we note that the run-time estimates for the steps in line 4 will depend on the number $t_k$ of marked items in the list $L_k$; however, it might well be that we \textit{don't know} how many marked items are there during any one step, and moreover this could be prohibitively time-consuming to compute. In such a situation we may be forced to estimate how many marked items there are, and this will introduce some error into our run-time estimates that we have to handle carefully. For instance, an unbiased estimate of the number of marked items in a list can give us a \emph{biased} estimate for the run-time of Grover search -- we discuss this and other considerations in more detail in Section~\ref{sec:simulating_quantum_algorithms}.

\subsection{Previous work}
There have been an increasing number of papers that perform precise resource estimates for a number of quantum algorithms, mostly with a focus on algorithms for simulation of physical systems~\cite{von2021quantum,lee2021even,whitfield2005quantum}. Others, such as~\cite{babbush2021focus}, have investigated what impact overheads such as error-correction might have on potential quantum speedups, in this case concluding that, at least in the near-term, quadratic or small polynomial speedups are unlikely to manifest in practice. Finally, Campbell et al.~\cite{campbell2019applying} performed a rigorous analysis of the potential speedups achievable by quantum algorithms for solving constraint satisfaction problems. They considered upper bounds on the run-times of both a naive application of Grover search as well as an optimized implementation of a more sophisticated quantum algorithm for backtracking due to Montanaro~\cite{montanaro2018quantum}, taking into account realistic properties of near-term as well as future hardware. They then compared these run-times to the performance of state-of-the-art classical algorithms in an effort to understand when quantum algorithms might provide a performance advantage, and what resources would be required for this. 

Our current work is similar in that we also use rigorous upper bounds on the complexities of our quantum sub-routines, although we are often more interested in \emph{expected} complexities. Moreover, we attempt to perform an analysis that is architecture independent, whereas Campbell et al.~were interested in hardware properties. We also consider quantum algorithms whose run-times cannot be analysed ahead of time, and which must be implemented, or simulated, in order to discover the speedups (or lack thereof) that they might achieve in practice.

There have also been works that aim to prove rigorous upper bounds on the complexities of various Grover search routines. For example, Zalka~\cite{zalka1999grover} performed a careful analysis to upper-bound the number of Grover iterations performed in the worst-case on a list with an unknown number of marked items. We make use of this result, and improve upon it, by extending the analysis to consider the \textit{expected} number of queries made by the algorithm, which requires substantially more effort to bound (tightly). Finally, we note that our analysis holds for all input sizes, whereas (as we understand it) Zalka's result applies only in the limit of large input size.

More recently, an arxiv preprint~\cite{stoudenmire2023grover} appeared that claimed that Grover's algorithm offers no quantum advantage. That paper explores the question of whether one would expect a speedup from applying Grover search in a different way than we do. We consider the speedup obtained by Grover vs.~it's classical counterpart - i.e.~brute-force search using the same query oracle that Grover has access to. They consider the question of whether Grover's algorithm itself can be classically simulated in practice whilst retaining the square-root speedup, which essentially boils down to studying the difficulty in classically simulating coherent calls to the oracle. It would be interesting (but far beyond the scope of this work) to add this approach to the toolbox when studying whether one can expect to obtain a speedup in practice via application of Grover-type quantum algorithms. The remainder of their paper considers the effects of noise on the performance of Grover's algorithm, which we explicitly avoided in our analysis.

\subsection*{Organization}

We begin in Section~\ref{sec:methodology} by explicitly describing an implementation of Grover search with an unknown number of marked items followed by an implementation of a quantum maximum finding routine. We then derive tight upper bounds, including all constants, for the expected and worst-case complexities of these quantum sub-routines. In Section~\ref{sec:simulating_quantum_algorithms}, we consider how to apply these bounds for a particular input without knowing ahead of time the parameters needed to compute them, and propose an estimation procedure that deals with this uncertainty. Finally, in Section~\ref{sec:maxsat}, we apply our methodology to the use-case of MAX-SAT, and present our numerical results.

\section{Query complexity bounds}
\label{sec:methodology}

In this section and the next we introduce the tools that form the backbone of our methodology for estimating the run-times of quantum algorithms, in the sense described above. The two main tools that we will require are: a set of rigorous upper bounds on the expected- and worst-case query complexities of Grover search with an unknown number of marked items and quantum maximum-finding (Section~\ref{sec:methodology}), and some technical results that allow us to estimate these complexities even when the exact number of marked items is unknown to us (Section~\ref{sec:simulating_quantum_algorithms}). 

As mentioned, our main quantum sub-routine will be a Grover search with an unknown number of marked items -- which we shall refer to as \textbf{QSearch} -- that can find and return a marked item from a list $L$ of length $|L|$ using an expected $O\left(\sqrt{\frac{|L|}{t}} \right)$ number of queries, when there are $t$ marked items in $L$:

\begin{lemma}[Grover's search with an unknown number of marked items~\cite{boyer1998tight}]
\label{lem:grover}
Let $L$ be a list of items, and $t$ the (unknown) number of `marked items'. Let $\mathcal{O}_g \ket{x_i}\ket{0} = \ket{x_i}\ket{g(x_i)}$ be an oracle that provides access to the Boolean function $g : [|L|] \rightarrow \{0,1\}$ that labels the items in the list. Then there exists a quantum algorithm {\bf QSearch($L,\epsilon$)} that finds and returns an index $i$ such that $g(x_i) = 1$ with probability at least $1-\epsilon$ if one exists and requires an expected number $O(\sqrt{N/t} \log(1/\epsilon))$ queries to $\mathcal{O}_g$ and $O(\sqrt{N/t}\log(N/\epsilon))$ other elementary operations. If no such $x_i$ exists, the algorithm confirms this and to do so requires $O(\sqrt{N}\log(1/\epsilon))$ queries to $\mathcal{O}_g$ and $O(\sqrt{N}\log(N/\epsilon))$ other elementary operations. 
% If we want to obtain \emph{worst-case} bounds, then there is a variant of the algorithm that behaves the same, and requires at most $O(\sqrt{N}\log(1/\epsilon))$ queries to $\mathcal{O}_g$ and $O(\sqrt{N}\log(N/\epsilon))$ other elementary operations. 
\end{lemma}
\noindent We will also find the following variant of Grover search useful in proving our bounds.
\begin{lemma}[Exact Grover search~\cite{hoyer2000arbitrary}]\label{lem:exact_grover}
    Let $L$ be a list of items, and $t>0$ the \emph{known} number of `marked items'. Let $\mathcal{O}_g \ket{x_i}\ket{0} = \ket{x_i}\ket{g(x_i)}$ be an oracle that provides access to the Boolean function $g : [|L|] \rightarrow \{0,1\}$ that labels the items in the list. Then there exists a quantum algorithm {\bf ExactQSearch($L,t$)} that finds and returns an index $i$ such that $g(x_i) = 1$ with \emph{certainty}. To do so, the algorithm makes $O(\sqrt{N/t})$ queries to $\mathcal{O}_g$ and $O(\sqrt{N}\log(N))$ other elementary operations.
\end{lemma}
\noindent Finally, we will make use of the quantum maximum-finding algorithm of~\cite{durr1996quantum}

\begin{lemma}[Quantum maximum-finding \cite{durr1996quantum}]
\label{lem:quantum_max}
Let $L$ be a list of items of length $|L|$, with each item in the list taking a value in an ordered set, to which we have coherent access in the form of a unitary that acts on basis states as 
	\[
		\mathcal{O}_L \ket{x}\ket{0} = \ket{x}\ket{L[x]}.
	\]
Then there exists a quantum algorithm {\bf{QMax}$(L, \epsilon)$} that will return $\argmax_x L[x]$ with probability at least $2/3$ using at most $O(\sqrt{|L|})$ queries to $\mathcal{O}_L$ (i.e. to the list $L$) and $O(\sqrt{|L|}\log |L|)$ elementary operations. By repeating the algorithm $\log(1/\epsilon)$ times, the probability of success can be amplified to $1 - \epsilon$.
\end{lemma}

In the sections that follow we carefully bound the expected and worst-case query complexities, including all constants, of \textbf{QSearch} on a list with an unknown number of marked items, and then of \textbf{QMax}. 
% Later, in Section~\ref{sec:numerics}, we use these bounds to precisely estimate the expected complexities of quantum algorithms that use \textbf{QSearch} and \textbf{QMax} as a sub-routine.
We consider two different implementations of \textbf{QSearch}, one that performs better in the expected case, the other better in the worst case. 
% Our reasons for focusing on query complexity over time complexity are explained in the Introduction, with the main one being that it allows for a more direct comparison to be made between the various quantum and classical algorithms, and does so in an architecture-independent way.

\

\noindent The implementation of \textbf{QSearch} uses both queries to the function $g$ (classical queries) and the oracle $\mO_g$ (quantum queries). Typically, the oracle $\mO_g$ can be constructed from a reversible classical circuit implementing $g$, in which case a single query to $\mO_g$ will generally require two queries to $g$ (to compute and uncompute garbage). When we refer to \emph{queries}, we will always mean queries to $g$ itself. A query to $\mO_g$ will then correspond to potentially multiple queries to $g$, and will be weighed with a constant $c_q$ denoting the number of queries made to $g$ per query to $\mO_g$. Generally speaking, and for the case of MAX-SAT discussed in Section~\ref{sec:maxsat}, $c_q = 2$ as mentioned above.

For notation, we will use $E$ to denote an expected number of queries to $g$, and $W$ for worst-case, with the name of the algorithm in question in the subscript. E.g. if we run \textbf{QSearch} on a list $L$ of length $|L|$ with $t$ marked items and a success probability of at least $1-\epsilon$, then the  number of queries will be denoted by $E_{\textbf{QSearch}}(|L|, t, \epsilon)$ in the expected case, and by $W_{\textbf{QSearch}}(|L|, \epsilon)$ in the worst-case. 

% Finally, we would like to remark that,depending on what kind of queries we count, classical and quantum queries might not contribute equally. \ido{Rephrase in terms of function $g$ in lemmas} For example, for the use case discussed in Section~\ref{sec:maxsat}, a query is a call to a function $\varphi$ that computes the weighted sum of satisfied clauses. In this case, a classical query corresponds to a single function call, whereas a quantum query refers to a query to the (reversbile) operation that computes and uncomputes $\varphi$, making a quantum query twice as expensive as a classical query. To take this discrepancy into account, we introduce a parameter $c_q$ that reflects how many classical queries correspond to a single quantum query\footnote{Which need not be equal to two in general.}. The derived expressions found in this paper for expected number of queries always refer to classical queries, which means that the number of calls to the quantum oracle will be weighted with $c_q$. In the use case discussed in Section~\ref{sec:maxsat}, $c_q = 2$.

\subsection{Expected query complexity of \textbf{QSearch}}
\label{sec:qsearch_expected_case}

Our implementation of \textbf{QSearch} is based on the implementation of Boyer et al.~\cite{boyer1998tight}, which takes as input a list $L$ of length $|L|$ and a unitary/oracle that gives coherent access to the function $g: L \rightarrow \{0,1\}$. Let $t = |\{x \in L: g(x) = 1\}|$ be the unknown number of marked items in $L$, which is assumed to be $\leq 3|L|/4$. We first introduce the algorithm $\qsb$, which consists of the following five steps:
\begin{enumerate}
    \item Set $\lambda = 6/5$, and initialize\footnote{Boyer et al.~initialize $\lambda=1$, which means they always start with one sample drawn uniformly at random. We choose to start with $\lambda=6/5$.} $m=\lambda$.
    \item Choose $j$ uniformly at random from the set of non-negative integers smaller than $m$.
    \item Apply $j$ Grover iterations to the uniform superposition of all items in the list
    \item Observe the list register. 
    \item If the observed item is marked, return it and exit; otherwise, set $m = \text{min}(\lambda m, \sqrt{|L|})$, and go back to step 2.
\end{enumerate}
Note that $\qsb$ will always find a marked item if there is one, but will run forever if there are no marked items. To obtain an algorithm with a finite stopping time one can add an appropriate time-out, in which case the algorithm has some probability of failing (reporting no marked items when there are in fact some). Boyer et al.~note that the case $t>3|L|/4$ can be disposed of in constant time using classical sampling. In order to simulate the behaviour of above algorithm numerically, we must include the classical sampling and time-out features explicitly.

In fact, in Appendix~\ref{app:improved_bounds} we show that it is not necessary to assume $0 < t \leq 3|L|/4$, and therefore the classical sampling part is optional. However, in case of many marked items, drawing classical samples is more efficient than applying Grover iterations, and for this reason we keep the classical sampling phase in our implementation of \textbf{QSearch} below, with the number of classical samples $N_{\text{samples}}$ as a hyperparameter. We discuss how to pick an optimal $N_{\text{samples}}$ in Section~\ref{sec:Optimizing_N_Samples}.

\subsubsection{Implementation}
We now give our implementation of $\textbf{QSearch}(L, N_{\text{samples}}, \epsilon)$  Here, $L$ is the list that contains marked and unmarked items, $N_{\text{samples}}$ is the number of classical samples we take and $1-\epsilon$ is the required lower bound for the success probability. We also define $\lambda = 6/5$ and $\alpha = 9.2$.

\begin{algorithm}[!htb]
  \caption{QSearch}
  \label{alg:QSearch}
  
\begin{algorithmic}[1]
    \Function{QSearch}{List $L$, integer $N_{\text{samples}}$, real number $\epsilon > 0$}
    \State $x \gets$ \Call{ClassicalSampling}{$L$, $N_{\text{samples}}$}
    \If{$x$ is marked}
        \State \Return $x$
    \EndIf 
    
    \State $N_\text{runs} \gets \ceil{\log_{3}(1/\epsilon)}$, $Q_{\text{max}} \gets \alpha \sqrt{|L|}$, $r \gets 0$
    \Comment{$\alpha = 9.2$}
    
    \While{$r < N_{\text{runs}}$}
        \State $m \gets \frac{6}{5}$, $Q_{\text{sum}} \gets 0$
        \State Sample a non-negative integer $j$ less than $m$ uniformly at random
        \While{$Q_{\text{sum}} + j \leq Q_{\text{max}}$}
            \State $y \gets$ \Call{GroverCycle}{$L$, $j$}
            \If{$y$ is marked}
                \State \Return $y$
                \Else 
                    \State $Q_{\text{sum}} \gets Q_{\text{sum}} + j + 1$
                    \Comment{Add $j+1$ quantum queries to total}
                    \State $m \gets \text{min}(\lambda m, \sqrt{|L|})$
                    \Comment{Update max number of iterations next cycle}
                    \State Sample a non-negative integer $j$ less than $m$ uniformly at random
                \EndIf
        \EndWhile
        \State $r \gets r+1$
    \EndWhile
    \State \Return No marked item found

    \EndFunction
\end{algorithmic}

\end{algorithm}
  
%%%%%%%%%%%%%%%%%%%%%%%%%%%%%%%%%%%%%%%%%%%%%%%%%%%%%%% 

\begin{algorithm}[!htb]
    \caption{Subroutines QSearch}
      \label{alg:Subroutines_QSearch}

\begin{algorithmic}[1]
    \Function{ClassicalSampling}{List $L$, integer $N_{\text{samples}}$}
        \State $k$ $\gets$ 0
        \While{$k < N_{\text{sample}}$}
            \State Sample an element $x$ from $L$ uniformly at random
            \If{$x$ is marked}
                \State \Return $x$
            \EndIf
            \State $k \gets k + 1$
        \EndWhile
        \State \Return No marked item found
        
    \EndFunction
\end{algorithmic}

%%%%%%%%%%%%%%%%%%%%%%%%%%%%%%%%%%%%%%%%%%%%%%%%%%%%%%%    

\begin{algorithmic}[1]
    \Function{GroverCycle}{List $L$, non-negative integer $j$}
        \State Prepare uniform superposition over all elements of $L$ 
        \State Do $j$ Grover iterations
        \State Measure list-index register
        \State \Return measurement outcome
    \EndFunction
\end{algorithmic}

\end{algorithm}

Our implementation works as follows. After sampling at most $N_{\text{samples}}$ items from $L$ classically, we execute $N_{\text{runs}}$ \emph{Grover runs}, where a single \emph{Grover run} is an application of $\qsb$ with a time-out, given by $Q_{\text{max}} = \alpha \sqrt{|L|}$ queries, i.e. lines 8 - 17 of Algorithm~\ref{alg:QSearch}. The number of runs depends on the desired success probability: $N_{\text{runs}} = \ceil{\log_3(1/\epsilon)}$. A single application of $\qsb$ with timeout consists of several \emph{Grover cycles}. That is, for a single run of $\qsb$ with timeout, we first initialize $m=\lambda$, and then repeatedly (i) pick an non-negative integer $j$ less than $m$, (ii) do $j$ Grover iterations and (iii) measure, and if we don't find a marked item we increase $m$ by a factor of $\lambda$. Steps (ii) - (iii) will be referred to as a Grover cycle, see Algorithm~\ref{alg:Subroutines_QSearch}. Finally, a single application of the Grover iterate\footnote{The unitary that first reflects through the unmarked states followed by a reflection through the uniform superposition.} will be referred to as a \emph{Grover iteration}.

% \begin{enumerate}
%     \item \emph{Classical sampling}
%     \begin{enumerate}
%         \item Set $k = 0$.
%         \item If $k > c$, go to step 2. Otherwise, pick an item from the list $L$ uniformly at random.
%         \item If the item is marked, return it, and exit. Otherwise increase $k$ by one and go back to (b). 
%     \end{enumerate}
%     \item \emph{Grover search} 
%     \begin{enumerate}
%         \item Set $i = 0$. 
%         \item If $i > r$, then exit and return `no marked item found'. Otherwise, initialize $m=1$, $\lambda = 6/5$ and $\Sigma = 0$.
%         \item If $\Sigma > J$, then increase $i$ by one and go back to (b).
%         \item Choose $j$ uniformly at random from the set of non-negative integers smaller than $m$, and update $\Sigma \leftarrow \Sigma + j$.
%         \item Apply $j$ Grover iterations to the uniform superposition of all items in the list
%         \item Observe the list-index register. 
%         \item If the observed item is marked, return it and exit; otherwise, set $m = \text{min}(\lambda m, \sqrt{|L|})$, and go back to (c).
%     \end{enumerate}
% \end{enumerate}

In Appendix~\ref{app:analysis_QSearch} (precisely Appendices~\ref{app:improved_bounds}-\ref{app:qsearch_exp_num_q}) we show that \textbf{QSearch}, as given by Algorithm~\ref{alg:QSearch}, has the properties stated in the lemma below. 

~\\

\begin{lemma}[Worst-case expected complexity of \textbf{QSearch}] 
\label{lem:QSearch}
Let $L$ be a list, $g: L \rightarrow \{0,1\}$ a Boolean function, $N_{\text{samples}}$ a non-negative integer and $\epsilon > 0$, and write $t = |g^{-1}(1)|$ for the (unknown) number of marked items of $L$. Then, $\textbf{QSearch}(L, N_{\text{samples}},\epsilon)$ as described by Algorithm~\ref{alg:QSearch} finds and returns an item $x \in L$ such that $g(x) = 1$ with probability at least $1-\epsilon$ if one exists using an expected number of queries to $g$ that is given by
\begin{equation}\label{eq:e_qsearch}
    E_{\textbf{QSearch}}(|L|,t,N_{\text{samples}},\epsilon) = \frac{|L|}{t}\left(1 - \left(1-\frac{t}{|L|}\right)^{N_{\text{samples}}}\right) + \left(1-\frac{t}{|L|}\right)^{N_{\text{samples}}} c_q E_{\text{Grover}}(|L|,t) \, ,
\end{equation}
where
\begin{equation}
    E_{\text{Grover}}(|L|,t)
    \leq F(|L|,t) \left(1 + \frac{1}{1 - \frac{F(|L|,t)}{\alpha \sqrt{|L|}}} \right) \, ,
\label{eq:QGrover_ubound_methodology}
\end{equation}
with
\begin{equation}
    F(|L|,t) =
    \begin{cases}
        \frac{9}{4}\frac{|L|}{\sqrt{(|L| -t )t}} + \left\lceil\log_{\frac{6}{5}}\left(\frac{|L|}{2\sqrt{(|L| -t )t}} \right) \right\rceil - 3 \leq \frac{\alpha \sqrt{L|}}{3 \sqrt{t}}   &\text{for} \quad 1 \leq t < \frac{|L|}{4} \\
        2.0344 &\text{for} \quad \frac{|L|}{4} \leq t \leq |L|. 
    \end{cases}
    \label{eq:F(L,t)}
\end{equation}
If no marked item exists, then the expected number of queries to $g$ equals the number of queries needed in the worst case (denoted by $W_{\textbf{QSearch}}(|L|, N_{\text{samples}},\epsilon)$), which is given by
\begin{equation}
    E_{\textbf{QSearch}}(|L|,0,N_{\text{samples}},\epsilon) = W_{\textbf{QSearch}}(|L|,N_{\text{samples}},\epsilon) \leq N_{\text{samples}} + \alpha c_q \ceil{\log_3(1/\epsilon)}) \sqrt{|L|} \, . \label{eq:worst_qsearch} \\
\end{equation} 
In the formulas above, $c_q$ is the number of queries to $g$ required to implement the oracle $\mO_g \ket{x}\ket{0} = \ket{x}\ket{g(x)}$, and $\alpha = 9.2$.
\end{lemma}

Note that our obtained expression for $E_{\textbf{QSearch}}(|L|,t,N_{\text{samples}},\epsilon)$ is actually independent of $\epsilon$ for $t>0$ because our upper bound for $E_{\text{Grover}}(|L|,t)$ is independent of $\epsilon$, which is a consequence of the fact that we don't have a lower bound on the failure probability; see Appendix~\ref{app:qsearch_exp_num_q} for details. Also, for the case $1 \leq t \leq |L|$, even though it might appear to be, the expected number of queries for the classical sampling part is actually not linear in $|L|$ since the term $\left(1 - \left(1-\frac{t}{|L|}\right)^{N_{\text{samples}}}\right)$ contains a factor of $\frac{t}{|L|}$.

\subsubsection{Optimizing $N_{\text{samples}}$}
\label{sec:Optimizing_N_Samples}

As mentioned in the beginning of Section~\ref{sec:qsearch_expected_case}, the number of classical samples we use for \textbf{QSearch} is a hyperparameter, and in this subsection we discuss how it can be optimized and set to improve the performance of the algorithm for different inputs.

Classical sampling requires fewer queries than Grover search does when a large fraction of the items is marked, whereas it is more efficient to not use classical sampling at all when a small number of them is marked. When a fraction $f$ of items are marked, then an expected $1/f$ classical queries will be required to find one. To determine when classical sampling is more efficient than quantum, we can compare this quantity with the number of queries made by Grover search. That is, given a list $L$ of size $L$, we want to find out for what value of $f$ we have
\begin{align}
\label{eqn: f_grover_vs_classical}
    \frac{1}{f} = E_{\text{Grover}}(|L|, f|L|).
\end{align}
The fraction where equality is attained we call $f_0$. When $f < f_0$, the Grover part requires fewer queries than classical sampling does, and therefore setting $N_{\text{samples}} > 0$ (i.e. turning the classical sampling part on) increases the query count compared to having $N_{\text{samples}} = 0$. We can numerically compute $f_0(|L|)$ for different values of $|L|$. We observe the following.
\begin{itemize}
    \item For $|L| \leq 260$, classical sampling always requires fewer queries.
    \item For $|L| \geq 260$, the value for $1/f_0$ that makes the rightmost inequality in Eq.~\eqref{eqn: f_grover_vs_classical} an equality is plotted as a function of $|L|$ in Fig.~\ref{fig:classical_vs_Grover_1f}.
    \begin{figure}[htb]
        \centering
        \begin{subfigure}{0.47\textwidth}
        \centering
        \includegraphics[scale=0.65]{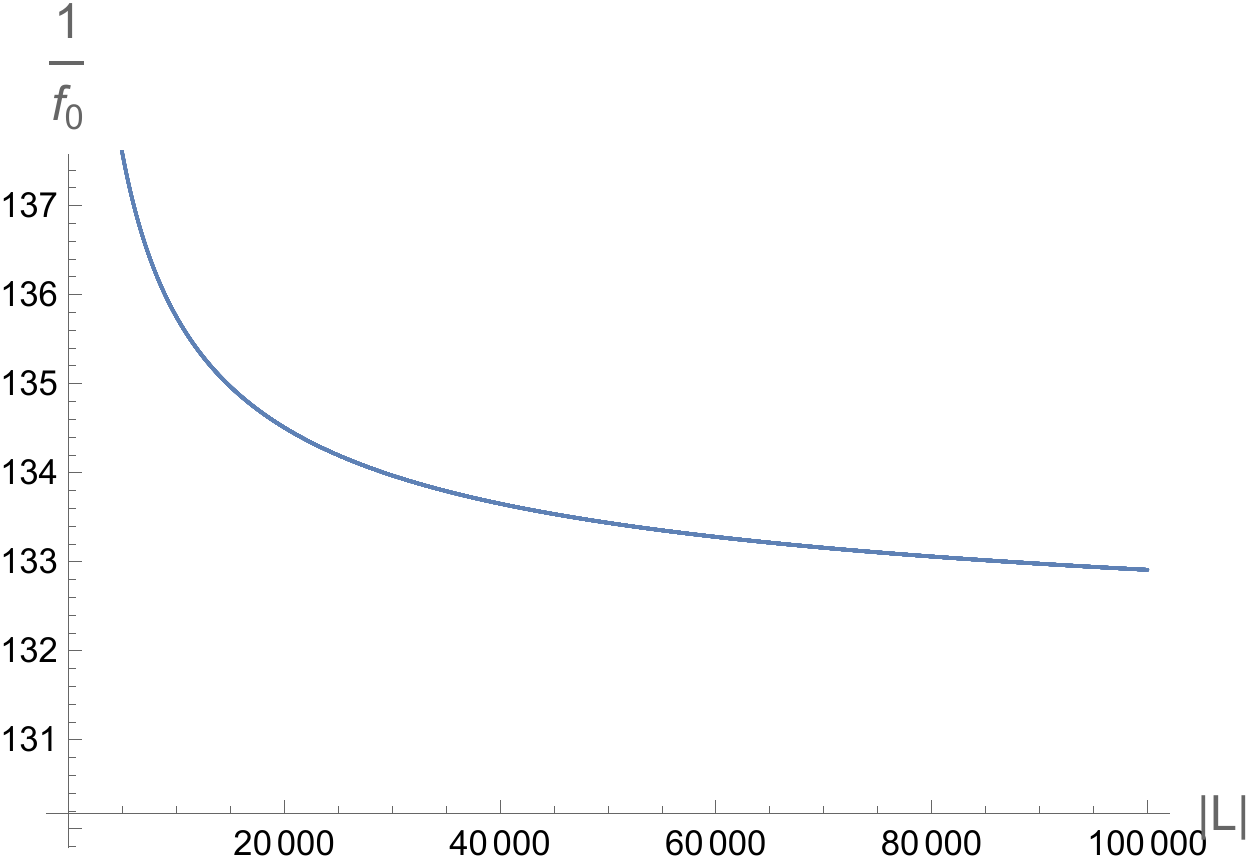}
        \caption{The value for $1/f_0$ as a function of the list length $|L|$ that marks the point beyond which, in expectation, Grover search requires fewer queries than sampling classically does. In the limit $|L| \rightarrow \infty$, there is a horizontal asymptote at $1/f_0 \rightarrow 131.665$.}
        \label{fig:classical_vs_Grover_1f}
        \end{subfigure}
        \hfill
        \begin{subfigure}{0.47\textwidth}
        \centering
        \includegraphics[scale=0.65]{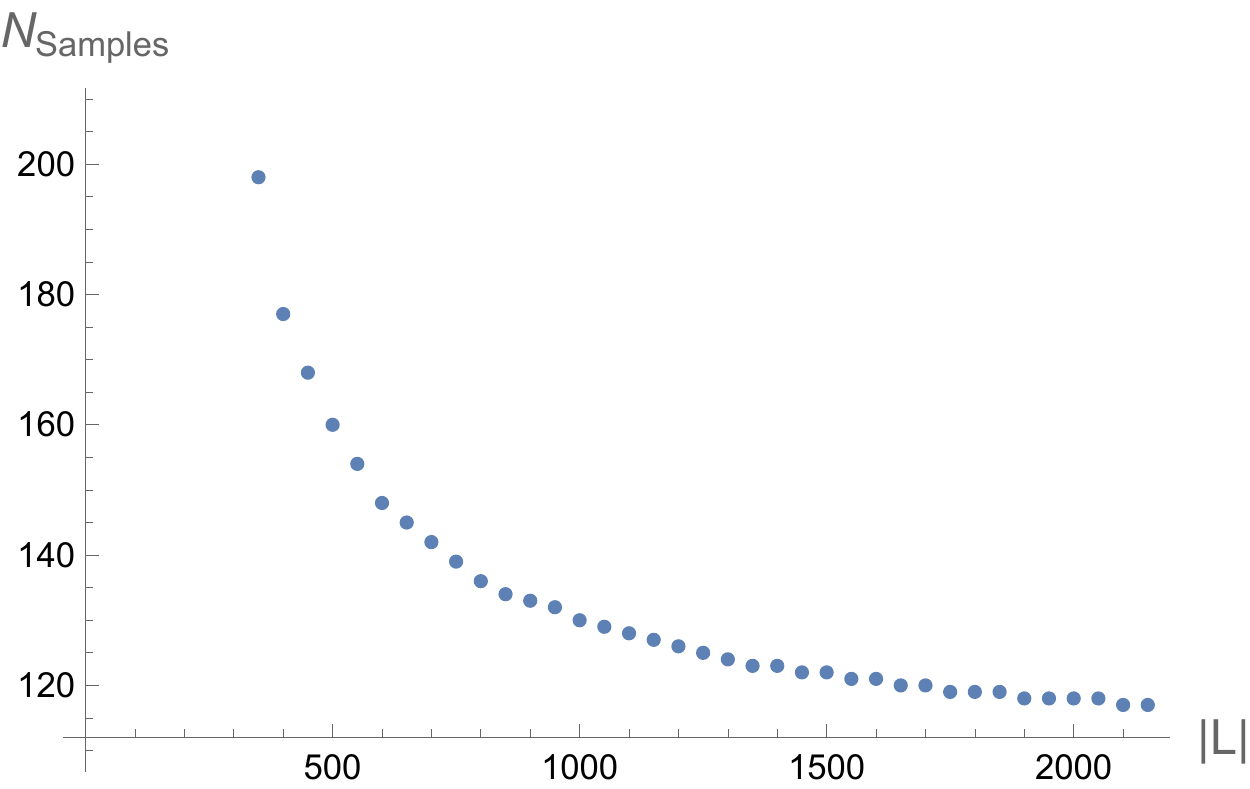}
        \caption{The optimal setting for $N_{\text{samples}}$ that minimizes the expected number of queries of $\textbf{QSearch}$ when we have no prior knowledge on the number of marked items $t$, i.e. every value of $t$ is equally likely.}
        \label{fig:optimal_NSamples}
        \end{subfigure}
    \end{figure}
\end{itemize}
In practice, we do not know what the fraction of marked items $f = t/|L|$ is. For certain algorithms, we might have some information about what $f$ can be, and in such cases this information can be leveraged to our advantage (see e.g.~\cite{cade2022community}). In case we have \textit{no prior knowledge} on the number of marked items at all, we can assume every value of $t$ is equally likely. In this case, the expected number of queries for $\textbf{QSearch}$ is given by 
\[
    \frac{1}{|L|} \sum_{t=1}^{|L|} E_{\textbf{QSearch}}(|L|, t, N_\text{samples}, \epsilon).
\]
Numerically minimizing\footnote{Recall that our upper bound for $ E_{\textbf{QSearch}}(|L|, t, N_{\text{samples}}, \epsilon)$ given by Eq.~\eqref{eq:e_qsearch} is independent of $\epsilon$ for $t>0$.} the expression above as function of $N_{\text{samples}}$ for small list sizes yields the graph in Fig.~\ref{fig:optimal_NSamples}, which gives an indication for what settings of $N_{\text{samples}}$ work well in practice. For our particular numerical simulations in Section~\ref{sec:numerics}, we set $N_\text{samples} = 130$.

\subsection{Worst-case query complexity of \textbf{QSearch}}
\label{sec:qsearch_worst_case}
If we make use of a slightly different implementation of \textbf{QSearch} described by Zalka~\cite{zalka1999grover}, we can obtain a tighter bound on its worst-case performance, which will be useful for some of our applications of \textbf{QSearch} where we only care about the worst case (since the expected complexity of this variant is actually worse than that of the $\textbf{QSearch}$ implementation described in the previous section). Unfortunately, the bounds given in~\cite{zalka1999grover} are only asymptotic and ignore, for example, extra constants arising from rounding integers, and so we briefly re-derive them below. The main upshot is that the dependence on the error becomes quadratically better than the usual implementation, which could end up being a significant improvement for our algorithms. To distinguish this implementation of \textbf{QSearch} from the one outlined in Section~\ref{sec:qsearch_expected_case}, we will refer to this quantum sub-routine as \ZQ.

As usual, let $L$ be the list of items over which we are searching, and suppose that $t$ of them are marked, and that we want to succeed in finding one if it exists with probability $\geq 1-\epsilon$. The algorithm, based on the one described in~\cite{zalka1999grover}, consists of the following steps:
\begin{enumerate} 
    % \item A classical first step (as in the implementation in Section~\ref{sec:qsearch_expected_case}) that rules out the case of a very large number of marked items. In particular we attempt to rule out $t \geq \frac{(a-1)|L|}{a}$ for some parameter $a > 1$ to be chosen. If we find a marked item in this step, we return it and stop. \ido{Don't forget me Chris}
    \item A preliminary step that checks for a \textit{small} number of marked items, by systematically ruling out $t=1, t=2, \dots, t=t_0$ for (for reasons made clear in Appendix~\ref{app:zalka_qsearch_proof}) $t_0 = \lceil \frac{\ln \epsilon}{2\ln(3/4)} \rceil$, by running \textit{exact} Grover search for each value of $t$. If this step finds a marked item, we return it and stop.
    \item A second step where (now with the knowledge from the first step that $t_0 < t$) we repeatedly choose an integer $j$ uniformly at random from the range $[0,\lceil\frac{\pi}{4} \sqrt{\frac{|L|}{t_0}}\rceil]$ and then run Grover search using $j$ iterations. This is done $2t_0$ times (which minimises the part of the complexity that depends on $|L|$). If a marked item is found during any run, we return it and stop. Otherwise, the algorithm returns `no marked item'.
\end{enumerate}

\paragraph{Run-time analysis}
The worst-case run-time is clearly when there are no marked items, and hence both steps above are run to the end. In Appendix~\ref{app:zalka_qsearch_proof} we prove the following lemma.

\begin{lemma}[worst-case complexity of \ZQ]
Let $L$ be a list of items, $g:L \rightarrow \{0,1\}$ a Boolean function and $\epsilon > 0$, and write $c_q$ for the number of queries to $g$ required to implement the oracle $\mO_g \ket{x}\ket{0} = \ket{x}\ket{g(x)}$. Then, with probability of failure at most $\epsilon$, \ZQ requires at most
\begin{equation}\label{eq:qsearch_max}
    W_{\textbf{QSearch}_{\text{Zalka}}}(|L|,\epsilon) := c_q \left(5\ceil*{\frac{\ln (1/\epsilon)}{2\ln(4/3)}} + \pi \sqrt{|L|} \sqrt{\ceil*{\frac{\ln (1/\epsilon)}{2\ln(4/3)}}} \right)\,
\end{equation}
queries to $g$ to find a marked item of $L$, or otherwise to report that there is none. 
\end{lemma}

\subsection{Quantum maximum finding \textbf{QMax}}\label{sec:maxfinding}

We use the quantum maximum finding algorithm from Ahuja and Kapoor~\cite{ahuja1999quantum}, described below. We then provide an improved analysis\footnote{We also believe that there is a slight error in the proof provided by~\cite{ahuja1999quantum}, which we fix in our proof.} for the expected number of queries made by their quantum maximum finding algorithm.
% \footnote{The quantum maximum finding algorithm also does not make use of classical queries, and therefore the factor of $c_q$ in our analysis. Note that classical queries can be included heuristically to make the algorithm perform better, since with high probability, the number of marked items is large in the early stages of the algorithm.}

The input of the algorithm is once again a list $L$, together with a function $R: L \rightarrow \R$ that assigns a value to each item. The output is the index  of an element of $L$ that maximises $R$. We assume we have coherent oracle access to the following marking function $f$ defined by
\begin{align}
\label{eq:qmax_oracle}
    f_i(j) = 
    \begin{cases}
        1 &\text{if} \quad  R(j) > R(i) \\
        0 &\text{otherwise} \, ,
    \end{cases}
\end{align}
i.e. access to unitaries $\mO_{f_i}$ that acts as
\begin{equation}
\label{eq:qmax_unitary}
    \mO_{f_i} \ket{j}\ket{0} = \ket{j}\ket{f_i(j)}\,.
\end{equation}

\subsubsection{Infinite-time algorithm}

To start with, we define a zero-error infinite-time\footnote{$\qmi$ will not stop running when it has found the maximum: it will continue to run indefinitely because it will be running $\qsb$ on a list with no marked items.} algorithm for finding the maximum. Afterwards, we will incorporate a time-out in the infinite algorithm, and use Markov's inequality to turn the infinite algorithm into a bounded-error algorithm that terminates in finitely many steps.

\begin{algorithm}[!htb]
  \caption{$\qmi$}
  \label{alg:QMax_infinity}
  
\begin{algorithmic}[1]
    \Function{$\qmi$}{List $L$}
    \State Choose $i \in L$ uniformly at random and set $y = R(i)$.
    \While{\textbf{True}}
        % \State Initialise the state $\ket{\psi} = \sum_{i=1}^{|L|} \ket{i} \ket{y}$.
        \State Apply $\qsb$ to the list $L$ with the marked items being  $f_y^{-1}(1)$.
        \State Update $y = R(j)$, where $j \in L$ is the item found by $\qsb$.
        %\State Measure the first register, set $y$ to be the output of the measurement.
    \EndWhile

    \State \Return $y$
    \EndFunction
\end{algorithmic}
\end{algorithm}

We say that $\qmi$ has \emph{found the maximum} when $y$ in Algorithm~\ref{alg:QMax_infinity} is equal to an item that maximises $R$. In Appendix~\ref{app:QMax-run-time}, we prove the lemma stated below.

\begin{restatable}{lemma}{lemmaEQMaxinf}[Expected complexity of $\qmi$]
\label{lem:EQMax-inf}
Let $L$ be a list of $|L|$ items. Then, the expected number of queries to any of the $f_i$ (as defined as in Eq.~\ref{eq:qmax_oracle}) required for $\qmi$ to find the maximum of $L$ is upper bounded by
\begin{equation}
    E_{\qmi}(|L|) \leq c_q \sum_{t=1}^{|L|-1} \frac{F(|L|,t)}{t+1} \, ,
    \label{eq:EQMax-inf}
\end{equation}
where $F(|L|,t)$ is defined by Eq.~\eqref{eq:F(L,t)}. Here, $c_q$ is the number of queries to $f_i$ required to implement the oracle $\mO_{f_i}$ (which we assume to be the same for all $i$).
\end{restatable}
In case we are interested in queries to $R$ rather than the $f_i$, then we note that the total number of queries to any of the $f_i$ combined is equal to the total number of queries to $R$ (since for each $f_i$ we need to compute $R(i)$ only once, and this we do anyway at the end of every Grover run to check if the found item is marked), \emph{except} at the very beginning, where in line 2 of Algorithm~\ref{alg:QMax_infinity} we need to compute $R(y)$. Therefore, the number of queries to $R$ is upper bounded by Eq.~\eqref{eq:EQMax-inf} plus one. Consequently, when running $\qmi$ a total of $T$ times, if we switch from upper bounding queries to any of the $f_i$ to queries to $R$, we need to add $T$ to our obtained bound for the former to obtain a bound for the latter. 

\

Next, we can use the bounds for the expected number of queries made by $\qsb$ to bound the expected number of queries for $\qmi$. Using the upper bound in Eq.~\eqref{eq:F(L,t)}, we have 
\begin{equation}
    F(|L|,t) \leq 
    \begin{cases}
        \frac{9.2\sqrt{|L|}}{3\sqrt{t}} &\text{if} \quad 1 \leq t < \frac{|L|}{4} \\
        2.0344  &\text{if} \quad \frac{|L|}{4} \leq t \leq |L| 
    \end{cases}
\end{equation}
Hence, we obtain
\begin{equation}
    E_{\qmi}(|L|) \leq c_q \left(\frac{9.2\sqrt{|L|}}{3} \sum_{t=1}^{\ceil{|L|/4}-1} \frac{1}{\sqrt{t}(t+1)} + 2.0344 \sum_{t=\ceil{|L|/4}}^{|L|-1} \frac{1}{t+1} \right)
    \, .
    \label{eq:EQMax_upper_bound}
\end{equation}
If the list $L$ is not too large, we can compute the above upper bound by evaluating the sum explicitly. If this computation becomes too time-consuming, we can also resort to bounds that are easier to evaluate. Two such bounds are derived in Appendix~\ref{app:QMax-Ubound}, and are given below. We have the following loose upper bound
\[
    E_{\qmi}(|L|) \leq c_q\left( 6.3505 \sqrt{|L|} + 2.8203 \right)\, ,
\]
as well as a tighter upper bound, given by
\begin{align*}
    E_{\qmi}(|L|) &\leq c_q \Bigg[ \frac{3 \sqrt{3}(1+\pi)}{4}\sqrt{|L|} + \frac{\ln(|L|/4)}{2 \ln(6/5)}\bigg(\ln(|L|/3) + \ln(|L|/4 + 1)  \bigg) \\ &\quad -2\ln(|L|/4)  + 5.3482 + \frac{\text{Li}_2(-\ceil{|L|/4}+1)}{2 \ln(6/5)} \Bigg]\, ,
\end{align*}
where $\text{Li}_2$ is Spence's function, also known as the dilogarithm. For the second bound the leading order term in $\sqrt{|L|}$,
\[
    c_q \frac{3 \sqrt{3}(1+\pi)}{4}\sqrt{|L|} \leq c_q 5.3801 \sqrt{|L|}\, ,
\]
has a smaller coefficient than the first upper bound has.

\subsubsection{Finite-time bounded-error algorithm}

Next, we can introduce a timeout $Q_{\tout} = 3 E_{\qmi}$ to make $\qmi$ a finite-time bounded-error algorithm. If $X$ is the random variable corresponding to the number of queries to $g$ made by $\qmi$ in order to find the maximum of the list $L$, then by Markov's inequality,
\[
    \Pr[X\geq Q_{\tout}] \leq \frac{E_{\qmi}}{Q_{\tout}} \leq \frac{1}{3} \,,
\]
resulting in a maximum-finding algorithm that finds the maximum with probability at least $\frac{2}{3}$ and uses an expected number of queries that is upper-bounded by $Q_{\tout}$.

We can further boost the success probability to $1-\epsilon $ by repeating the above $\log_{3}(1/\epsilon)$ times, and picking the largest element (with respect to the function $R$) out of all repetitions, to obtain the algorithm $\textbf{QMax}$ that succeeds with probability at least $\epsilon$ and makes at most $\ceil{\log_{3}(1/\epsilon)}Q_{\tout}$ queries in expectation.

\begin{corollary}[Expected complexity of \textbf{QMax}]
Let $L$ be a list of items of length $|L|$. Let $f_i$ be the marking functions as defined in Eq.~\ref{eq:qmax_oracle}. Then the expected number of queries to $f_i$ (for any $i$) required for \textbf{QMax} to find the maximum of $L$ with success probability at least $1-\epsilon$ is $\ceil{\log_3(1/\epsilon)}3 E_{\qmi}(|L|)$, where $ E_{\qmi}(|L|)$ is given by Eq.~\eqref{eq:EQMax-inf}.
\end{corollary}
Using the upper bounds derived for $E_{\qmi}$ derived above, it is sufficient to choose 
\[
    Q_{\tout} \geq c_q \left( 19.0515 \sqrt{|L|} + 8.4609 \right) \,
\]
using the loose upper bound, or
\begin{align*}
    Q_{\tout} &\geq c_q \Bigg[ \frac{9 \sqrt{3}(1+\pi)}{4}\sqrt{|L|} + \frac{3\ln(|L|/4)}{2 \ln(6/5)}\bigg(\ln(|L|/3) + \ln(|L|/4 + 1)  \bigg) -6\ln(|L|/4)  \\
    &\quad + 16.0466 + \frac{3\text{Li}_2(-\ceil{|L|/4}+1)}{2 \ln(6/5)} \Bigg]
\end{align*}
using the tight upper bound.

\section{Estimating complexities under uncertainty}
\label{sec:simulating_quantum_algorithms}

To estimate the query complexities of \textbf{QSearch} and \textbf{QMax}, we can use the bounds derived in the previous sections for their expected- and worst-case complexities. These bounds take as input a list $L$, the desired success probability of the sub-routine, and in case of \textbf{QSearch} also the number $t$ of marked items in $L$.  However, the number of marked items in the list will not be known ahead of time, and moreover could be computationally time-consuming to compute classically, especially for very large inputs, which will likely be the ones for which we want to estimate the run-times of quantum algorithms.

Additionally, for algorithms of form of Algorithm~\ref{alg:general} that make use of repeated calls to \textbf{QSearch} or \textbf{QMax}, we would like that with high probability \textit{every} call to either \textbf{QSearch} or \textbf{QMax} succeeds, which will require boosting their success probabilities to something (inversely) proportional the number of times they are run. However, the number of times each sub-routine is called will often not be known until the algorithm has finished executing, and therefore we will require a reliable upper bound $T$ to the number steps, i.e. the number of times such calls are made.

In this section, we discuss how to deal with both quantities. First, in Section~\ref{sec:estimating_marked_items}, we discuss how to use a sampling procedure to estimate the number of marked items using sampling, and consider the extra complications that arise when we use such estimated values to compute (bounds on) the complexities of algorithms. Next, in Section~\ref{sec:unknown_number_of_steps}, we discuss how the total number of steps affects the accuracy of both \textbf{QSearch} and \textbf{QMax}, as well as the accuracy of the estimates for the expected number of queries made by these quantum routines as obtained through the sampling procedure discussed in Section~\ref{sec:estimating_marked_items}.

% With this in mind we will often be driven to estimate both of these quantities. The number of steps taken by the algorithm will usually be quite context-dependent, and we address this point for our use-case specifically in Section \ido{REF: do we still need to know (an upper bound on) the number of moves?}. The other, the number of marked items, can be estimated using sampling. However, this can lead to some extra complications when we use such estimated values to compute (bounds on) the complexities of the algorithms, and we address this concern below in Section~\ref{sec:estimating_marked_items}.

% which will generally depend on how many times the sub-routine is called during all $T$ steps of the algorithm.  the number of times each sub-routine is called will often not be known until the algorithm has finished executing. Moreover,

\subsection{Estimating the number of marked items}
\label{sec:estimating_marked_items}

% As we have already discussed, our main quantum tool is a Grover search with an unknown number of marked items, which we name \textbf{QSearch}. If we are searching over a list of $|L|$ items, and there are $t$ of them marked, then we have already established bounds on the expected number of oracle queries made by \textbf{QSearch}. 

To use the bounds on the number of queries made by \textbf{QSearch} derived in Section~\ref{sec:qsearch_expected_case}, we need to know how many marked items there are in the list given as input to a \textbf{QSearch} call. For sufficiently small lists, we can count the number of marked items exactly at reasonably little computational cost. For longer lists this will become very time consuming, leading to exceedingly slow simulations. When determining the number of marked items exactly becomes infeasible, we can instead estimate the number of marked items. We can do this by counting the number of samples $l$ we need to draw on average before we find a marked vertex. 

For a list $L$ with $t$ marked items, the probability that an element of $L$ chosen uniformly at random is marked is $f = t/|L|$. Consequently, the probability that we find a marked item after randomly choosing (with replacement) $k \in \bbN$ elements of $L$ (the first $k-1$ elements not being marked) is given by
\begin{equation}
    \Pr[l = k] = (1-f)^{k-1} f \, 
    \label{eq:geo-distr}
\end{equation}
(i.e~a geometric distribution with parameter $f$). We write $l\sim \text{Geo}(f)$ to denote a random variable sampled according to such a distribution, and throughout this section, when we take the expectation value over $l$ it is implied that we do this over the geometric distribution, i.e.:
\[
    \E[X(l)] = \E_{l \sim \text{Geo}(f)}[X(l)] \, 
\]
for any function $X: \bbN \rightarrow \R$. By sampling $l \sim \text{Geo}(f)$, we obtain an unbiased estimate of 
\[
    \E[l] = 1/f = |L|/t\, ,
\]
which we can use to approximate the expected number of queries made by \textbf{QSearch} if it were run on $L$. 

\ 

\noindent To start we focus on our upper bound for the expected number of (quantum) queries ($E_{\text{Grover}}$) to the oracle $\mO_g$ made by \textbf{QSearch}. In order to estimate (an upper bound for) $E_{\text{Grover}}(|L|,t)$, we would like an estimator $E_{\text{Grover}}^{\text{estimator}}$ such that the procedure (i) sample $l \sim \text{Geo}(f)$, and then (ii) plug the result into our expression for $E_{\text{Grover}}^{\text{estimator}}(l)$ gives, in expectation (over $l$), an upper bound to $E_{\text{Grover}}(|L|,t)$; i.e.~we want $\E[E_{\text{Grover}}^{\text{estimator}}(l)] \geq  E_{\text{Grover}}(|L|,t)$.

A naive attempt at constructing $E_{\text{Grover}}^{\text{estimator}}$ would be to take our upper bound on $E_{\text{Grover}}(|L|,t)$ from Eq.~\eqref{eq:QGrover_ubound_methodology} and in this expression replace $1/t$ by $l/|L|$. However, from Eq.~\eqref{eq:F(L,t)}, we observe that $F(|L|,t)$ contains concave functions like the square-root and the logarithm, which, by Jensen's inequality satisfy
\begin{equation*}
    \E [\sqrt{l}] \leq \sqrt{\E[l]} = \sqrt{1/f} \quad \text{and} \quad 
    \E [\log_{\frac{6}{5}}{l}] \leq \log_{\frac{6}{5}}{\E[l]} = \log_{\frac{6}{5}}{1/f}\, .
\end{equation*}
As a consequence, the procedure outlined above (in expectation over $l$) does \emph{not} give an upper bound to the expected number of queries; instead we obtain a biased estimator that \emph{underestimates} our upper bound for $E_{\text{Grover}}$. Note that the issue of concavity will arise in any approach that tries to simulate a Grover search on an unknown number of marked items by using classical sampling to estimate the fraction of marked items. 

\ 

\noindent We now discuss how to deal with the concavity of the square-root and logarithm, and then describe an estimator that always upper bounds $E_{\text{Grover}}$ in expectation.

\paragraph{Upper bound for square-root and log estimates}

% In general, if one wants to estimate a Grover-like speed-up during of some search subroutine of a heuristic algorithm, one has to know the fraction the good elements encompass of the entire search space. Let this fraction at iteration $k$ be $p_k$.  Exactly calculating $p_k$ might not be feasible if the search space is too large, and instead one could opt for random sampling from the search space to estimate this fraction. Let $X$ be the random variable corresponding to the total amount of samples taken until a marked item is found. If we first estimate the fraction by $1/X_k$ one creates a \textit{biased} estimator for a prediction of the quadratic speed-up Grover gives, as can be shown by Jensen's inequality
% \begin{equation}
%     \mathbb{E} [\sqrt{X_k}] \leq \sqrt{\mathbb{E}[X_k]} = \sqrt{1/p_k},
% \end{equation}
% since $\sqrt{X_k}$ is concave for $X_k\geq 1$ and $\mathbb{E}[|X_k|] < \infty$. This means that this method on average underestimates the amount of quantum queries. 
We prove the following two lemmas in Appendices~\ref{app:proof_square_root} and~\ref{app:proof_logarithm}:
\begin{restatable}{lemma}{jensens}\label{lem:sampling_constant}
For a random variable $X$ geometrically distributed with parameter $f$, there exists a constant $d_1=4/\pi\approx  1.273$, such that
\[
 \sqrt{\mathbb{E}[X]} \leq\mathbb{E} [\sqrt{d_1 X}]
\]
for all $f\in(0,1]$. 
\end{restatable}

% \begin{restatable}{lemma}{jensens}\label{lem:sampling_constant}
% There exists a constant $d_1=2/\sqrt{\pi}\approx  1.128 $ such that  we have that 
% \[
%  \sqrt{1/x} \leq d_1 h(x)
% \]
% for all $x\in(0,1]$. This bound has relative error at most \ido{Do we need this?}
% \[
%  \frac{|d_1 h(x) - \sqrt{1/x}|}{\sqrt{1/x}} \leq 0.128,
% \]
% attained when $x\rightarrow 1$. For small $x$, the relative error goes to zero.
% \end{restatable}
% In Appendix~\ref{app:proof_logarithm} we prove a similar statement for the logarithm.

\begin{restatable}{lemma}{jensenslog}\label{lem:sampling_constant_log}
For a random variable $X$ geometrically distributed with parameter $f$, there exists a constant $d_2 = e^{\gamma}\approx 1.781$, where $\gamma$ is the  Euler–Mascheroni constant, such that
\[
 \log({\mathbb{E}[X]}) \leq\mathbb{E} [\log{(d_2 X)}]
\]
for all $f\in(0,1]$. 
\end{restatable}

\noindent Moreover, the proof of Lemma~\ref{lem:sampling_constant} also tells us that in the interesting regime, i.e. where the fraction $f$ is small, the relative error between our upper bound and the actual expected value $\sqrt{1/f}$ goes to zero.  

Using the two lemma's above, in Appendix~\ref{app:additive_constant}, we show that the following estimator 
\begin{equation}
    E_{\text{Grover}}^{\text{estimator}}(l) :=
    -1.1272 + \frac{1.7850}{\sqrt{|L|}} + \frac{1.2991}{\sqrt{|L|}} l + \left(5.1962 - \frac{2.5064}{\sqrt{|L|}}\right) \frac{2\sqrt{l}}{\sqrt{\pi}} +  \frac{5}{4} \log_{\frac{6}{5}}\left(e^\gamma l\right)
    \label{eq:QGrover_estimator_final}
\end{equation}
upper bounds $E_{\text{Grover}}$ in expectation for all $1 \leq t \leq |L|$ (or $1 \leq \frac{1}{f} \leq |L|$):
\[
    \E[E_{\text{Grover}}^{\text{estimator}}(l)] \geq E_{\text{Grover}}(|L|,t) \, ,
\]
where the expectation is taken over the geometric distribution $l \sim \text{Geo}(f)$, with $f = t / |L|$.

\paragraph{Estimation procedure}
Using the estimator $E_{\text{Grover}}^{\text{estimator}}$, we can construct an estimator that in expectation upper bounds the expected number of queries $E_{\textbf{QSearch}}$ given by Eq.~\eqref{eq:e_qsearch}. To do so, we require the following two functions on the set of positive integers: $h_1,h_2: \bbN \rightarrow \R$
\[
    h_1(l) = \min(l, N_{\text{samples}}) \, ,
\]
and 
\[
    h_2(l) = \begin{cases}
        0 \qquad l\leq N_{\text{samples}} \\
        1 \qquad l>N_{\text{samples}}
    \end{cases}
\]
Given $h_1$ and $h_2$, in Appendix~\ref{app:Estimator-QSearch} we prove the following lemma.

\begin{lemma}
\label{lem:Estimator_QSearch}
The estimator
\begin{equation}
    H(l) = h_1(l) + h_2(l) c_q E_{\text{Grover}}^{\text{estimator}}(l)
    \label{eq:estimator-qsearch}
\end{equation}
upper bounds $E_{\textbf{QSearch}}$ in expectation:
\[
    \E[H(l)] \geq E_{\textbf{QSearch}}(|L|, t, N_{\text{samples}}, \epsilon)
\]
for $1 \leq t \leq |L|$ and for all\footnote{Note that the derived expression for $ E_{\textbf{QSearch}}(|L|, t, N_{\text{samples}},\epsilon)$ is independent of $\epsilon$ for $t \geq 1$.} $\epsilon > 0$, where the expectation value is taken over the geometric distribution $l\sim \text{Geo}(f)$, with $f = t/|L|$.
\end{lemma}

The above estimator applies to the situation where there is at least one marked item. However in case $t=0$, in order to determine $l$ we keep drawing samples (with replacement) indefinitely. To make sure our algorithm terminates in finite time, we need to set a maximum $l_{\max}$ such that, if $l = l_{\max}$, we conclude that $t=0$ and we stop the sampling procedure. The particular choice of $l_{\max}$ will depend on our tolerance for falsely detecting no marked items.

\

% \noindent In practice, we use the following procedure to obtain \ido{run-time} estimates for \textbf{QSearch} when we don't know the number of marked items $t$. 

% \noindent Using the estimator from Eq.~\eqref{eq:estimator-qsearch}, we obtain the procedure $\textbf{Estimate}_{\textbf{QSearch}}(L,N_{\text{samples}},l_{\max},\epsilon)$ for estimating (an upper bound to) the expected number of queries to $g$ made by \textbf{QSearch}, which works as follows. 

\noindent In practice, we use the procedure $\textbf{Estimate}_{\textbf{QSearch}}(L,N_{\text{samples}},\delta,\epsilon)$ described in Algorithm~\ref{alg:EsimateQsearch} for estimating (an upper bound to) the expected number of queries to $g$ made by \textbf{QSearch}. Recall that $c_q$ is the number of queries to $g$ required to implement the oracle $\mO_g$.

\begin{algorithm}[!htb]
    \caption{}
      \label{alg:EsimateQsearch}

\begin{algorithmic}[1]
    \Function{$\textbf{Estimate}_{\textbf{QSearch}}$}{List $L$, integer $N_{\text{samples}}$, failure probabilities $\delta$ and $\epsilon$}
        \State $l_{\max}  \gets  \lceil \frac{|L|}{\delta} \rceil$.
        % \State $\text{Found} \gets \text{False}$
        % \State 
        % \State $l \gets 0$
        % \While{Not Found}
        %     \State Sample an item from $L$
        % \EndWhile
        \State Draw samples uniformly at random (with replacement) from $L$ until either finding a marked item, or making $l_{\max}$ samples. Let $l$ be the number of samples taken.
        
        \If{$l \leq N_{\text{samples}}$}
            \State Then a marked item would have been found classically, in which case
            \[
            E \gets l\, .
            \]
        \ElsIf{$N_{\text{samples}} <l \leq l_{\max}$}
            \State Then the marked item would not have been found classically, and some Grover iterations would have been performed. In such a case,
            \begin{align*}
        E &\gets N_{\text{samples}} + c_q \Bigg[ -1.1272 + \frac{1.7850}{\sqrt{|L|}} + \frac{1.2991}{\sqrt{|L|}} l \\
        &\quad + \left(5.1962 - \frac{2.5064}{\sqrt{|L|}}\right) \frac{2\sqrt{l}}{\sqrt{\pi}} +  \frac{5}{4} \log_{\frac{6}{5}}\left(e^\gamma l\right) \Bigg] \, .
            \end{align*}
        \ElsIf{$l > l_{\max}$}
            We conclude $t=0$, and therefore, by Eq.~\eqref{eq:worst_qsearch},
    \[
        E \gets N_{\text{samples}} + 9.2 c_q \ceil{\log_3(1/\epsilon)}) \sqrt{|L|} \, .
    \]
        \EndIf
            
        \State \Return E
        
    \EndFunction
\end{algorithmic}
\end{algorithm}

\begin{lemma} \label{lem:estimate_procedure}
Let $L$ be a list of items, $g:L \rightarrow \{0,1\}$ a Boolean function and $\epsilon,\delta > 0$. Write $t = |g^{-1}(1)|$ for the (unknown) number of marked items of $L$. Then, with probability at least $1-\delta$, the procedure $\textbf{Estimate}_{\textbf{QSearch}}(L,N_{\text{samples}},\delta,\epsilon)$ of Algorithm~\ref{alg:EsimateQsearch}, gives an upper bound to the expected number of queries made by $\textbf{QSearch}(L, N_{\text{samples}}, \epsilon)$, in expectation over $l\sim \text{Geo}(f)$, where $f = t/|L|$.
\end{lemma}

\begin{proof}
This follows from Lemma~\ref{lem:Estimator_QSearch}, except now we have to take into account the possibility that we sample $l=l_{\max}$ even when $t\geq 1$. Recall that in Algorithm~\ref{alg:EsimateQsearch} we set $l_{\max} = \lceil \frac{|L|}{\delta} \rceil$. Assuming the worst-case of $f = 1/|L|$, by Markov's inequality, this can happen with probability at most
\[
    \Pr[l \geq l_{\max}] \leq \frac{\E[l]}{l_{\max}} \leq \frac{|L|}{l_{\max}} = \delta \, .
\]
\end{proof}
This lemma gives a very rough upper bound on the failure probability as it does not take the fraction of marked elements into account. In practice we found that setting $\delta = \frac{1}{100}$ worked well. Also note that, to get more accurate estimates, rather than sampling $l$ and plugging the expression into $H$, we can also sample $l$ multiple times and take the sample average of all the corresponding $H$-values. This will however require a somewhat more elaborate failure probability analysis in case some of the sampled $l$'s are equal to $l_{\max}$.

\subsection{Unknown number of steps}
\label{sec:unknown_number_of_steps}

When running an algorithm of the form of Algorithm~\ref{alg:general}, we require that all calls to the quantum subroutines \textbf{QSearch} or \textbf{QMax} to succeed in order to guarantee that the final algorithm worked correctly. This requires boosting their success probabilities, and, in order to do so, we need to know how many times each one is called, which in our case means knowing how many steps the overall algorithm will take. In case of heuristic algorithms, the number of steps is usually not known. However, it is not uncommon for heuristic algorithms that their \textit{typical} behaviour on certain practical problem instances is known -- which is the case for, for example, MAX-SAT and community detection~\cite{cade2022community}.

In case nothing is known about the number of steps, then instead we can (i) guess an upper bound to the number of steps, (ii) run the classical algorithm that emulates the quantum algorithm, and (iii) check retro-actively if indeed we required fewer steps than the guessed upper bound. If our guess was too low, then optionally we can increase it and repeat.

Suppose that $T$ is such a (guessed) upper bound to the total number of steps of an algorithm of the form of Algorithm~\ref{alg:general}, meaning that we call \textbf{QSearch} or \textbf{QMax} at most $T$ times. By the union bound, given a desired probability of failure of at most $\epsilon_{\text{total}}$, the accuracies of the individual subroutines $\epsilon_{\text{subroutine}}$ should be set such that
\[
    (1-\epsilon_{\text{subroutine}})^T \geq 1 - \epsilon_{\text{total}}\, ,
\]
meaning it is sufficient to choose
\[
    \epsilon_{\text{subroutine}} \leq 1 - (1-\epsilon_{\text{total}})^{1/T}\, .
\]
The above formula in fact holds for the accuracy of the quantum subroutines (denoted by $\epsilon$ in the sections above), as well as the parameter $\delta$ (which determines $l_{\max}$) in Lemma~\ref{lem:estimate_procedure} for the procedure $\textbf{Estimate}_{\textbf{QSearch}}(L,N_{\text{samples}},\delta,\epsilon)$, since this estimation procedure is called once per step of the (classical simulation of the) algorithm.

% we have that the inverse of the success probability -- i.e. the expected number of times one has to run the subroutine -- depends on the number of marked items $t_k=n f_k$, with $n$ the total amount of nodes and $f_k$ the fraction of `good nodes' over the total amount of nodes at iteration $k$ of the algorithm, in the following way:
% \[
% f(t_k) = \frac{1}{1-\frac{1}{3 \sqrt{t_k}}}.  
% \]
% The sampling procedure estimates $t_k \sim n / X_k$, where $X_k$ is again the random variable describing the amount of samples we have to take before we find a marked element. Hence, we obtain
% \[
% f(X_k) \sim \frac{1}{1-\frac{1}{3} \sqrt{\frac{X_k}{n}}},  
% \]
% which is a concave function in $X_k$. Therefore, again by Jensen's inequality, we have that sampling on average underestimates the expected number of subroutine repeats and therefore is an unfair estimator of $\mathbb{E}[f(X_k)]$. Trivial bounds for $\mathbb{E}[f(X_k)]$ would be 
% \[
% 1\leq \mathbb{E}[f(X_k)] \leq 3/2,
% \]
% for which the lower bound is attained in the case when $t_k \sim n$ (which is when (almost) all elements are marked) and the upper bound when $t_k$ is very small.

% \jordi{to do: Numerical benchmarks: comparing the sampling procedure to exact fraction calculations}

\section{Use-case: \maxsat }
\label{sec:maxsat}
In this section, we take the tools developed in Sections~\ref{sec:methodology} and~\ref{sec:simulating_quantum_algorithms} and apply them to a particular heuristic, called a \emph{hill-climber}, for finding (approximate) solutions to Boolean satisfiability problems. The algorithm discussed is of the sort that it admits a quantum speedup that is of the form of Algorithm~\ref{alg:general}. This section achieves the modest goal of numerically confirming| that the proposed framework works, but is limited in its depth; for a more comprehensive and detailed numerical study (of a different computational problem) we would like to refer the reader to~\cite{cade2022community} (see also Section~\ref{sec:concwork}).

\subsection{Propositional Boolean Satisfiability ($k$-{\sc SAT})} 
$k$-{\sc sat} is a fundamental problem in computer science and artificial intelligence, in which we ask whether a satisfying assignment exists for a given Boolean formula in conjunctive normal form, with the property that each clause contains at most $k$ literals. Whilst $k$-{\sc sat} is an example of a \emph{decision problem}, \maxsat is an \emph{optimization problem} that generalizes $k$-{\sc sat}: it is the problem of determining the maximum number of clauses, that can be made true by an assignment of truth values to the variables of the formula. Let $\mathbf{x} \in \{0,1\}^n$ be bit strings of length $n$,  $C = \{ C_i\}_{i=1}^m$ be a set of $m$ clauses, which each act on at most $k$ literals, and $W = \{w_i\}_{i=1}^m \subseteq \mathbb{R}^m$ a set of weights. The goal of \maxsat  is to solve
\[
\max_{\mathbf{x}} \varphi(\mathbf{x}),
\]
where $\varphi(\mathbf{x}) = \sum_{i=1}^m w_i C_i(\mathbf{x})$. This problem is NP-hard for any $k\geq 2$. 

A straightforward heuristic for solving \maxsat instances is based on \textit{hill-climbing}: the general idea is to start with some initial bit string, and then look for incremental improvements in the direct neighbourhood of this given bit string. This process is repeated iteratively until it has converged to some local maximum or the maximum number of iterations is reached.  Hill-climbing belongs to the family of \emph{local search} methods in mathematical optimization. Local search heuristics have been widely studied for SAT and MAX-SAT (see Ref.~\cite{Sttzle2001ARO} for an extensive review on local search methods) and are also yet still being studied: see Refs.~\cite{Alkasem2021StochasticLS,cai2017decimation} for some more recent works.

For \maxsat , we define the $d$-level neighbourhood $\mathcal{N}_d(\mathbf{x})$ of some bit string $\mathbf{x}$ as the set of all other bit strings that differ from $\mathbf{x}$ in \emph{at most} $d$ bit flips. The total size of this space is given by
\[
|\mathcal{N}_d(\mathbf{x})| = \sum_{i=1}^d {n \choose i} =  \mO(n^{d}).
\]
For our hill climber heuristic, we either consider a \emph{simple} hill climber, which greedily moves to an arbitrary neighbouring bit string with a strictly larger objective function value, and the \emph{steep ascent} hill climber, which computes $\varphi$ on all bit strings in the neighbourhood of the current bit string and picks the one that maximises the increase in $\varphi$ (assuming its objective function value is strictly larger than that of the current bit string. 

If we write $T$ for the number of moves made by either the simple or the steep ascent hill climber (which in general will require different number of steps depending on the problem instance, and in case of the simple hill climber also on the internal randomness of the algorithm) the worst-case time complexities of both algorithms have similar mathematical expressions, given by
\begin{align}
   \sum_{t \in [T]} \mO(n^d ) = \mO(T n^d)\, ,
   \label{eq:TCCSDHC}
\end{align}
because the per step complexities have the same worst-case upper bounds. However, in practice the \emph{expected} run-time of the simple hill climber depends on the instance and the current state of the algorithm: the more bit strings in its neighbourhood increase the objective function value the faster it completes its local search step in expectation. If, at step $t$ of the algorithm, we write $f_{d,t}$ for the fraction of the number of bit strings in $\mathcal{N}_d (\mathbf{x}_t)$ for which $\varphi$ assumes a value larger than $\varphi(x_t)$, i.e. 
\[
    f_{d,t} = \frac{\left|\{x \in \mathcal{N}_d(x_t):\varphi(x) > \varphi(x_t) \}\right|}{|\mathcal{N}_d(x_t)|} \, ,
\]
then we can bound the expected number of steps for the simple hill climber by
\begin{align}
\sum_{t\in[T]} \mO\left(\frac{1}{f_{d,t}}\right).
\end{align}

\subsection{Quantum heuristics for \maxsat }\label{sec:quantum_maxsat_algo}
Both variants of the hill climber search routines lend themselves to be sped up easily by Grover implementations.

To start with, given a bit string $\textbf{y}$, we define the function 
\[
    f_{\textbf{y}}(\textbf{x}) = 
    \begin{cases}
        1 &\text{if} \quad  \varphi(\textbf{x}) > \varphi(\textbf{y}) \\
        0 &\text{otherwise} \, .
    \end{cases}
\]
We assume that, for every bit string $\textbf{y} \in \{0,1\}^n$, we have oracle access to $\mO_{f_{\textbf{y}}}$.

\begin{lemma}[Simple quantum hill-climber] Let $\varphi$ be a \maxsat  instance on $n$ variables, and assume oracle access to each of the $\mO_{f_{\textbf{y}}}$ as described above. Then there exists a quantum algorithm \textbf{Simple quantum hill-climber} that with probability $\geq 2/3$, behaves identically to a classical simple hill climber and requires at most an expected number
\begin{align}
\sum_{t\in[T]} \tilde{\mO}\left(\sqrt{\frac{1}{f_{d,k}}}\right)
\end{align}
calls to $\varphi$.
\label{lem:simple_quantum_hill_climber}
\end{lemma}
\begin{proof}
We pick an initial bit string as we would with the classical simple hill climber. Next, suppose that in step $t$ of the algorithm our current best bit string is $\textbf{x}_t$. Here, we replace the local search step of a $d$-level neighbourhood in the aforementioned classical simple hill-climber by a single call to \textbf{QSearch} using the oracle $O_{f_{\textbf{x}_t}}$. Writing $f_{d,k}$ for the fraction of neighbours of $\textbf{x}_t$ for which the objective function value is strictly larger than $\varphi(\textbf{x}_t)$, by Lemma~\ref{lem:grover} we require at most an expected number $O(\sqrt{\frac{1}{f_{d,t}}} \log(1/\epsilon))$ queries to $\mathcal{O}_{f_{\textbf{x}_t}}$ and $O(\sqrt{\frac{1}{f_{d,t}}}\log(n^d/\epsilon))$ other elementary operations to find such a candidate $\bm{x}_{t+1}$ with probability at least $1-\epsilon$. If we set $\epsilon = 1-\sqrt[T]{\frac{2}{3}}$, our overall success probability will be at least $(1-\epsilon)^T = 2/3$, as required.

Since each query to any of the $O_{f_{\textbf{x}_t}}$'s requires $O(1)$ queries to $\varphi$, the lemma statement follows.
\end{proof}
\begin{lemma}[Steep quantum hill-climber] Let $\varphi$ be a \maxsat  instance on $n$ variables, and assume oracle access to each of the $\mO_{f_{\textbf{y}}}$ as described above. Then there exists a quantum algorithm \textbf{Steep quantum hill-climber} that with probability $\geq 2/3$, behaves identically to some classical steep hill climber and requires at most an expected number
\begin{align}
\sum_{t\in[T]} \tilde{\mO}(n^{d/2})
\end{align}
calls to $\varphi$.
\end{lemma}
\begin{proof}
The proof is similar to the proof of Lemma~\ref{lem:simple_quantum_hill_climber}, except that in this case, instead of the classical maximum finding routine, we make use of the quantum subroutine \textbf{QMax} of Lemma~\ref{lem:quantum_max}, which also requires access to each of the $O_{f_{\textbf{x}_t}}$'s. We can find the maximum in $\mO(\sqrt{n^d})\log(1/\epsilon)$ queries to each of the $O_{f_{\textbf{x}_t}}$'s with probability $\geq 1-\epsilon$. We set $\epsilon$ in the same way as we did for the simple quantum hill climber to get the desired success probability.
\end{proof}

\subsection{Numerics}\label{sec:numericsss}

In this section, we describe our numerical implementations of the classical and quantum versions of the steep and simple hill climbers. We then compare the expected number of queries for the quantum and classical versions when applied to typical problem instances of \maxsat using the method developed in Sections~\ref{sec:methodology} and~\ref{sec:simulating_quantum_algorithms}.

\subsubsection{Algorithmic implementations}

\paragraph{Classical hill climbers}
Both classical algorithms are allowed to sample \textit{without} replacement when searching for a good (or the best) element. Therefore we have that, in the case of a simple hill climber, the number of samples $X_{d,t}$ when searching over a list of $|\mathcal{N}_d(\mathbf{x}_t)|$ elements at step $t$, of which a fraction of $f_{d,t}$ are `good elements', has an expected value given by
\begin{align}
    \mathbb{E}[X_{d,t}] = \frac{|\mathcal{N}_d(\mathbf{x})|+1}{|\mathcal{N}_d(\mathbf{x})| f_{d,t}+1}.
\end{align}
For the steep hill climber, the classical expected number of queries at every step is always equal to $|\mathcal{N}_d(\mathbf{x})|$.

\paragraph{Quantum hill climbers}
For the quantum algorithms, we set the desired failure probability $\epsilon$ for the entire algorithm to be at most $10^{-5}$, which can be achieved by setting the accuracy per step to $\epsilon/T$, with $T$ the maximum total number of steps. Empirically, we found that $T=n$ provides a very loose upper bound on the total number of steps taken by the algorithm. Note that the value of $T$ could be optimised more thoroughly -- this leads to a smaller total number of queries needed in the quantum setting -- but we leave this for now as this is beyond the main goal of this case study. 

\

\noindent For the simple hill climber we use two implementations, one that calculates the number of marked elements $t$ (in this case marked elements correspond to possible moves that increase the objective function value) exactly at every step, and one that acquires only an estimate of this via sampling. 

The exact implementation keeps track of the list of all marked elements at every step, which allows us to use our sharper bounds from Lemma~\ref{lem:QSearch} in Section~\ref{sec:qsearch_expected_case} to upper bound the expected number of queries made by \textbf{QSearch} at each step of the algorithm. From this list of marked elements, we select an element at random and use that as an update step for the classical simulation. 

The sampling algorithm, just like its classical counterpart, instead samples in search of elements that give an increase in the cost function. When it finds one, use the number of tries $l$ it took to find a marked item as input to estimate (an upper bound) to the expected number of queries \textbf{QSearch} would have made, as described in Section~\ref{sec:simulating_quantum_algorithms}, to estimate the run-time of the quantum algorithm. This procedure is just an implementation of Algorithm~\ref{alg:EsimateQsearch} for the case of \maxsat. 

The steep ascent hill climber also keeps track of the complete list of marked items at every step. From this it selects the item with the maximal function value increase. It uses our bounds from Section~\ref{sec:maxfinding} to attain estimates of the expected number of queries \textbf{QMax} would have made for every step, in order to estimate the run-time of the entire algorithm.

\subsubsection{Numerical implementation}
We write the problem as a matrix multiplication problem and use numpy to solve it, which allows for larger instances to be tested. The assignment of truth values $x \in \{0,1\}^n$ is written as a vector $\tilde{x} \in \{-1,1\}^n$ where $-1$ is assigned to variables that are false and $1$ to those that are true. The clauses C can be written in a similar fashion, $\tilde{C}_i \in \{-1,0,1\}^n$, where $-1$ is assigned to the negated variables, $0$ is assigned to the variables that are not in the clause, and $1$ to those that should be true according to the clause. We construct a matrix $A$ with the $\tilde{C}_i$'s as rows. This matrix has an efficient sparse representation since most of it's entries are $0$. The objective function $\phi(x)$ becomes the following:
\[
    \tilde{\phi}(\tilde{x}) = W^T \left(\left\lceil\frac{A \tilde{x} + k}{2 k}\right\rceil\right)\,,
\]
where $W^T$ is row vector containing the weights for each clause and $k$ is the number of variables per clause. The addition and division of $k$ is elements-wise, while $A\tilde{x}$ is matrix vector multiplication. Note that $ - k \leq A \tilde{x} \leq k$, where the left-hand inequality is only attained when all variables are incorrectly assigned. In that case the ceiling function returns a 0 and in all other cases it returns a 1, as required. In all numerical simulations $d$ (which determines the level of the neighbourhood considered) is set to $1$. 

\paragraph{Sampling implementation}
At every step, the sampling algorithm samples up to $d$ (that determines the size of the neighbourhood of $\tilde{x}$ that the hill climber algorithm can search over; in our case $d = 1$) indices of $\tilde{x}$ and flips their value by multiplying by $-1$. It then calculates the objective function $\tilde{\phi}(\tilde{x})$ and accepts the changes if the cost increased, and rejects otherwise. This is repeated until the algorithm rejects $10n$ times\footnote{This is the value of $l_{\max} = 10n$ in Algorithm~\ref{alg:EsimateQsearch}.} in a row, at which point we assume that the algorithm has converged.

\paragraph{Exact implementation}
At every step, the exact implementation calculates the cost increase of every possible change to $\tilde{x}$. This is done by constructing the matrix
\[
    B(\tilde{x}) = \sum_{ij} x_i (1 - 2 \delta_{ij}),
\]
which consists of $n$ copies of $\tilde{x}$ as columns, where we multiply all diagonal elements by $-1$. This represents all the possible changes of $\tilde{x}$ at a single step. Now we can use $\tilde{\phi}(B(\tilde{x}))$ to calculate the cost of all possible changes simultaneously. This gives a vector $\tilde{y}$ containing the cost value of all the $n$ possible new configurations of $\tilde{x}$ (assuming $d = 1$). These values are compared to the old cost value and those that give a positive increase are saved in a list of marked elements. We consider those variables for which a change (being multiplied by $-1$) incurs a positive increase in the objective function as marked, the all other variables as unmarked. The size of this list gives the exact value of $t$, the number of marked variables. The exact implementation of simple quantum hill-climber selects one marked element from the list at random. The steep quantum hill-climber selects the marked element with the highest cost value. 

\paragraph{Data structure}
\label{sec:data_structure}
The exact implementation is feasible due to the fact that we use matrix multiplication to calculate the cost values. However, it can still be quite slow for larger instances. To remedy this, to an extent, we add an extra data structure that keeps track of the list of marked variables, rather than reconstruct it at every step. To do so we use the fact that any update is in some sense local. Let the $i$'th index of $\tilde{x}$ be the index that is updated. Then there is a subset of clauses $\{C_{j}|C_{i j} \neq 0\}$ (rows of $A$ where the $i$'th index of the clause is not zero). These are the only clauses that can change from being satisfied to not satisfied, or from not satisfied to satisfied, by changing the $i$'th index of $\tilde{x}$. Not all variables of $\tilde{x}$ are contained in these clauses (only $k$ variables get assigned a non-zero value in a clause). Exactly those variables that are, can change from being marked to not and visa versa. Hence we only need to consider this subset of variables when updating the list of marked variables. This severely reduces the computational cost of keeping track of marked items. As it turns out, this is efficient enough to avoid running-time limitations, but instead makes memory limitations the bottleneck. 

\subsubsection{Results}\label{sec:numerics}
Here we present our results for estimating the run-times of the two quantum algorithms described previously. Specifically, we estimate the number of queries to any  of the marking functions\footnote{We could have also chosen to count queries to $\varphi$ instead. Note that, after finding $\textbf{x}_t$ at step $t$, we know $\varphi(\textbf{x}_t)$ from the checking part of $\qsb$ (used as a subroutine for both \textbf{QSearch} and \textbf{QMax}), so every query to $f_{\textbf{x}_t}$ corresponds one query to $\varphi$. The difference between counting queries to the marking functions versus counting queries to $\varphi$ occurs at initialisation, where we need one extra query to $\varphi$ to compute the function value of the initial bit string that is not taken into account when counting queries to the marking functions. Hence, for a total of $T$ calls to either \textbf{QSearch} or \textbf{QMax}, the number of queries to $\varphi$ equals the number of queries to the marking functions plus $T$. This relationship holds for both the classical and quantum query counts. In our comparison, we chose to compare queries to marking functions, because this is where the speedup manifests itself.} $f_{\textbf{y}}$ from Section~\ref{sec:quantum_maxsat_algo} by applying the bounds obtained in Section~\ref{sec:methodology}. We set $c_q = 2$, since the quantum algorithms for \maxsat make queries to an oracle $\mathcal{O}_{f_{\mathbf{y}}}$, which requires 2 queries to $f_{\mathbf{y}}$ to implement.

\begin{figure}[htb]
    \centering
    \includegraphics[width=0.8\linewidth]{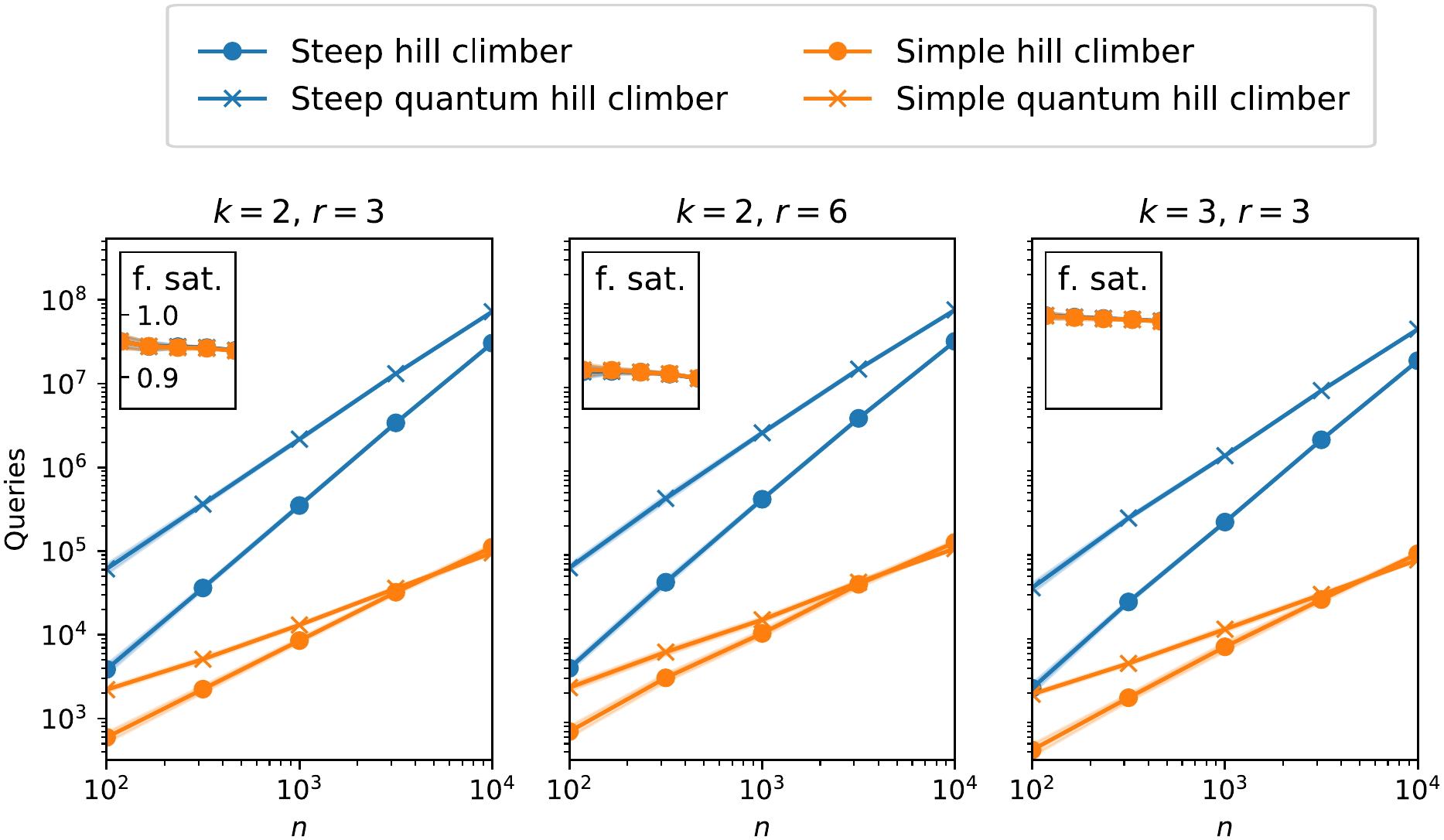}
    \caption{Numerical results for the query counts on randomly generated \maxsat  instances, with $n$ variables and $m=r n$ weighted (uniformly random between 0 and 1) clauses, of the proposed classical and quantum algorithms that implement a hill climber search routine. All hill climbers only consider a $d=1$ level neighbourhood, so the local search space size is at any step $k$ equal to $n$. The horizontal axis indicates the total amount of variables $n$ and the vertical axis the amount of queries made to any of the marking functions. Each data point corresponds to the average over 10 randomly generated instances and the shaded area represents one standard deviation. In every sub-figure the inlet plots the fractional number of weighted satisfied clauses, defined as $\varphi(\bm{x^*})/W$, where $W=\sum_{i\in[m]} w_i$ is total weight on the $m$ clauses and $\varphi(\bm{x^*})$ the objective function value for the obtained solution $\bm{x^*}$, the $x$-axis of the subplots is the number of nodes $n$. The blue and orange lines in the sub-figures are overlapping, this shows that the quality of the solutions found is comparable for the different algorithms. The classical algorithms are indicated by a `$\bullet'$, the respective quantum algorithms by a `$\times$'.}
    \label{fig:SAT_plots_overview}
\end{figure}

\

\noindent We tested our algorithms on different instances of \maxsat  to see what kind of speed-ups can be attained on average-case instances. The instances were generated using a random assignment of $k$ variables per clause.  Figure~\ref{fig:SAT_plots_overview} shows the average number of queries made by our classical and quantum algorithms. There, $n$ is the number of variables, $k$ the number of variables per clause, $r$ is multiplied by $n$ to get the number of clauses $m = rn$. We observe that the behaviour is very similar amongst the different parameter choices in the random \maxsat  generation. We find, as one might expect, that both quantum versions of the steep and simple hill climbers achieve better asymptotic \textit{scaling} when compared to their classical counterparts: here better asymptotic scaling means that we expect that the polynomial which describes the number of queries made to the cost function has a lower degree for the quantum algorithm than it has for the classical one. This is indicated by the difference in slope of the plots in Figure~\ref{fig:SAT_plots_overview}, as the number of queries against the problem size is plotted on a log-log scale and thus gives information about the degree of this polynomial, provided $n$ is large enough. On the contrary, only the simple quantum hill climber is able to also beat the classical algorithm in terms the of \textit{absolute} number of queries for the problem sizes considered (since only in this case the plot corresponding to the number of quantum queries goes below the classical one). However, since it achieves better scaling, we expect that for slightly larger $n$ (larger than $10^4$) the steep quantum hill climber will also start to beat the classical counterpart on average, as one would expect. The interesting point here is that, even for a fairly simple model that only takes query counts into consideration, the problem sizes need to already be quite large in order to achieve a quantum speedup.

Table~\ref{tab:summary_results} shows the empirically observed asymptotic scaling behaviour of our algorithms. By taking a linear fit in the log-log plot we can estimate the scaling exponents of the different algorithms. In Table~\ref{tab:summary_results} we show the relative speedup of our quantum algorithms compared to their classical counterpart. We see that a part of the theoretical speedup is lost. This is likely due to a combination of the fact that the theoretical speedup is a per-step speedup that does not affect the total number of steps taken, only the number of queries required for each individual step, and the fact that on relatively small instances the extra overhead required to run the quantum algorithms is significant. 
% We expect that the difference in performance will become even more profound if one would consider larger values of $d$. \marten{I don't really expect that actually, only for the steep hill climber, but mostly because we make the problem unnecessary hard }\chris{Maybe we won't say it then :)}
\begin{figure}[htb]
    \centering
    \includegraphics[width=\linewidth]{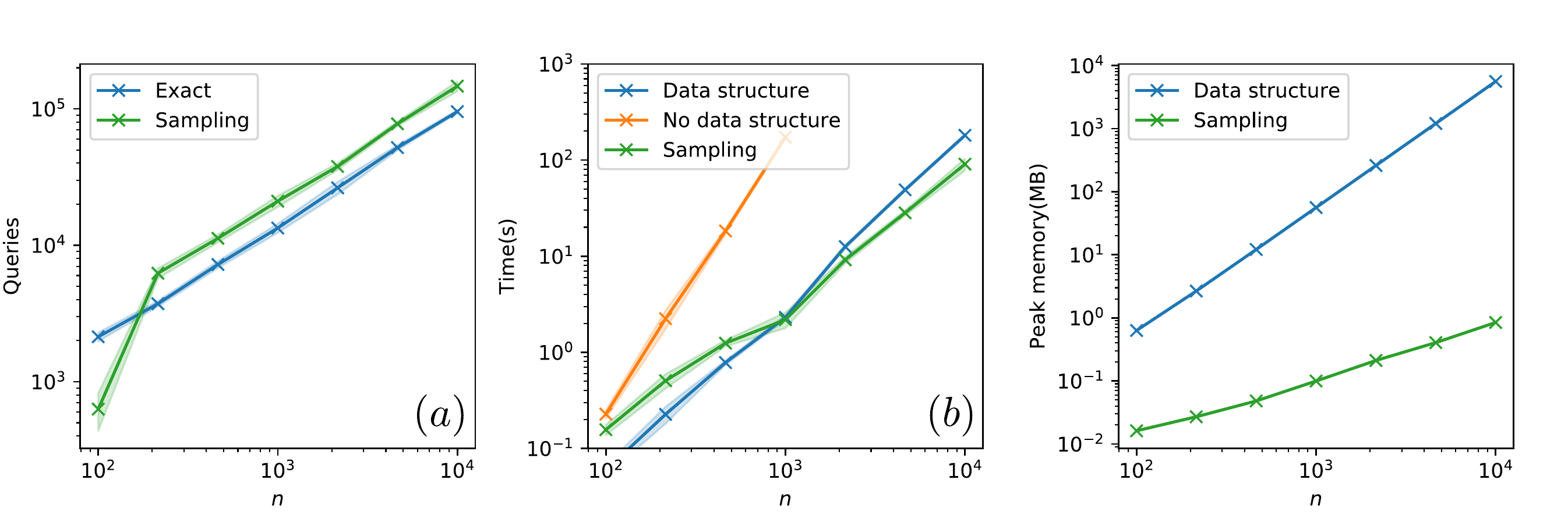}
    \caption{Several numerical results for sampling and exact methods for the simple hill climber on random 2-SAT instances with $r = 3n$ clauses, with $N_\text{samples}=130$  a): average query count, b) average running times and c) peak memory usage. The first plot shows a comparison of query count between the sampling and exact method. Note how for the smallest value of $n$ our sampling method fails to yield a proper upper bound: this is due to the fact that $N_\text{samples} > n$, which results in the fact that with high probability we fail to turn on Grover at all, and as a consequence we underestimate the expectation value (since the contribution from Grover to the expectation value is high). The second plot shows a run-time comparison between the exact (with and without data structure) and the sampling method. The third plot shows a comparison of peak memory usage between the exact method (with data structure), and the sampling method. The exact method without data-structure is not shown in the third plot as it has the same memory usage as the sampling method. }
    \label{fig:sampling_vs_exact}
\end{figure}

\

\noindent As discussed in Section~\ref{sec:simulating_quantum_algorithms}, when instances become too large we cannot use an exact method anymore to keep tack of the number of marked items. In Figure~\ref{fig:sampling_vs_exact} we show a comparison between the exact methods and our introduced sampling method for estimating an upper bound on the expected number of queries, for $N_\text{samples}=130$. We find that our estimation method provides a decent upper bound on the number of queries in expectation. For the exact methods we consider two different implementations to acquire the necessary information for calculating the expected number of queries at every step. The first one runs over the entire search space at every step acquiring the number of marked items. The second one uses the data-structure --- described in Section~\ref{sec:data_structure} --- that exploits the locality of the instances to update the fraction of `good elements' in the neighborhood of a given bit string.

Regarding the run-times of our classical simulations, Figure~\ref{fig:sampling_vs_exact} shows that both sampling and the data-structure method considerably outperform the exact implementation that runs over the entire search space. However, the extra added data structure comes at the cost of additional memory requirements, which become the bottleneck as we consider problems at a larger scale. Therefore, for instances where $n > 10^4$, we are limited to the usage of the sampling methods to obtain results. Finally, we note that the data structure method is very context-specific (i.e. here the data structure is specific to \maxsat) and might not always be possible, whereas the estimation method is applicable generally.

\begin{table}[htb]
\centering
\begin{tabular}{l|l|l|l|l}
 & \makecell{Classical query \\ complexity \\ per iteration $\tau$} 
 & \makecell{Quantum query \\ complexity \\ per iteration $\tau$}  & \makecell{Absolute\\ speed-up \\ observed?} 
 & \makecell{Empirically observed\\ range of polynomial \\ speed-ups} \\ 
 \hline
%  \makecell{\textbf{Simple hill climber}} & $ \mO\left( \frac{1}{f_{\tau}} \right)$ & -- & -- \\
%  \makecell{\textbf{Steep hill climber}} & $ \mO\left( |L| \right)$ & -- & -- \\\hline
\makecell{\textbf{Simple hill climber}}
& $ \mO\left( \frac{1}{f_{\tau}} \right)$ 
& $\tilde{\mO}(\sqrt{\frac{1}{f_{\tau}}} )$  
& \ccgg Yes  
& \ccgg $1.45$-$1.72$ \\ 
\makecell{\textbf{Steep hill climber}}
& $ \mO\left( |L| \right)$ 
& $\tilde{\mO}( \sqrt{|L|} )$  
& \ccrr No  
& \ccgg $1.38$-$1.60$ 
\end{tabular}
\caption{Shown are the theoretically obtained per-iteration complexities of our algorithms compared to their empirically observed speedups across the entire algorithm. Here 'absolute speedup' refers to the quantum algorithm making fewer (estimated) queries than the classical algorithm on the datasets that we considered. The numbers shown in the rightmost column measure the speedup achieved by the quantum algorithm: these are obtained by a linear weighted fit on the plots of Figure~\ref{fig:SAT_plots_overview}, which gives the scaling exponent of the expected query counts as a function of the problem size; the number in the table is the classical exponent divided by the corresponding quantum exponent. The numbers are larger than one in all cases, indicating a (modest) quantum speedup. The maximum speedup that can be obtained is 2, since that would correspond to the full quadratic per-step speedup manifesting across the entire run-time. Note that the steep hill-climber would likely also achieve an absolute speed-up if we considered slightly larger problem instances, as it achieves a better scaling than it's classical counterpart.}
\label{tab:summary_results}
\end{table}

\subsection{Summary of results}
Our main findings can be summarised as follows.
\begin{itemize}
    \item The quantum hill climbers obtain favourable scaling compared to their classical counterparts, but only one of them (the simple hill climber) obtained an absolute (query) speedup compared to its classical counterparts.
    \item Our estimation procedure gave reliable upper bounds on the complexities of the quantum algorithms as compared to an exact procedure, confirming our theoretical analysis from Section~\ref{sec:simulating_quantum_algorithms}.
    \item Our estimation procedure significantly decreased the computational cost of obtaining run-time estimates in the way considered in this paper. An exact approach that made use of a particular data structure yielded similar results, however it added large memory costs, and such an approach will always be very context-specific and sometimes not possible at all. 
    \item Classical heuristic algorithms tend to work by making many fast-to-compute but small updates to minimize the cost function, a structure that does not lend itself to significant quantum speedups.
\end{itemize}

\bibliography{main.bib}
\bibliographystyle{plain}

\paragraph{Acknowledgements}
We would like to thank Harry Buhrman for suggesting the idea of estimating input-dependent run-times, Ian Marshall for helpful discussions, and Quinten Tupker for help with the proof of Lemma~\ref{lem:sampling_constant_log}. The numerics of Section~\ref{sec:numerics} were carried out on the Dutch national e-infrastructure with the support of the SURF Cooperative.

\paragraph{Funding}
CC was supported by QuantERA project QuantAlgo 680-91-034, with further funding provided by QuSoft and CWI. MF and JW were supported by the Dutch Ministry of Economic Affairs and Climate Policy (EZK), as part of the Quantum Delta NL programme. IN was supported by the DisQover project:  a  collaboration between QuSoft and ABN AMRO, and recieved funding from ABN AMRO and CWI. 

\appendix

\section{Detailed analysis of \textbf{QSearch}}
\label{app:analysis_QSearch}

In this section of the appendix we give details to support the bounds on the success probability and expected number of queries made by our implementation of \textbf{QSearch} given in Section~\ref{sec:qsearch_expected_case}.

As mentioned in the beginning of Section~\ref{sec:methodology}, queries to the quantum oracle $\mO_g$ come with a weight of $c_q$ relative to the classical queries to $g$. In Sections~\ref{app:improved_bounds} and \ref{app:qsearch_succes_prob}, a query refers to a query to $\mO_g$. Only in Section~\ref{app:qsearch_exp_num_q} will we include the classical queries, and then the queries to the quantum oracle will by multiplied by an extra factor of $c_q$ (where needed) in the expressions obtained for the expected number of queries to $g$ for \textbf{QSearch}.

\subsection{Improved bounds}
\label{app:improved_bounds}

To start with, let us briefly go over the original analysis of Boyer et al.~\cite{boyer1998tight}, and improve some of the bounds where we can. The analysis in this subsection applies to $\qsb$.

Suppose we have a list $L$ with $t$ marked items, and let $\theta$ be such that
\[
    \sin^2(\theta) = t/|L| \, .
\]
Moreover, let
\[
    m_t = \frac{1}{\sin(2\theta)} = \frac{|L|}{2 \sqrt{(|L|-t)t}} \,.
\]
Now, the following lemma provides a lower bound to the success probability of finding a marked item with a single Grover run.

\begin{lemma}
\label{lem:Pm_Boyer}
(Lemma 2 from~\cite{boyer1998tight}). Suppose we have a list $L$ with $t$ marked items, and let $\theta$ be such that $\sin^2(\theta) = t/|L|$, $m \in \bbN_{>0}$ an arbitrary positive integer. Then, the probability $P_m$ of finding a marked element after doing $j$ Grover iterations, where $j$ is a non-negative integer smaller than $m$ chosen uniformly at random, is given by
\[  
    P_m = \frac{1}{2} - \frac{\sin(4m\theta)}{4m\sin(2 \theta)}\, .
\]
Consequently, if $m \geq m_t$, then $P_m \geq \frac{1}{4}$.
\end{lemma}

According to the algorithm description of $\qsb$, we initialize $m = \lambda$ and $\lambda = 6/5$. After every run, we multiply $m$ by $\lambda$. The moment $m > m_t$, we reach the so-called \emph{critical stage}. As Boyer et al.~observe, because of Lemma~\ref{lem:Pm_Boyer}, once in the critical stage every run has probability of at least $1/4$  to find a marked item, and this lower bound can be used to upper bound the expected number of Grover iterations required to find a marked item.

\

The issue with the requirement $m \geq m_t$ is that, when $\theta \rightarrow \pi/2$, $m_t \rightarrow \infty$. For this reason, Boyer et al.~exclude the regime of $\theta$ close to $\pi/2$ -- which corresponds to the case of many marked items -- by classical sampling. However, as we show below, in the regime $|L|/4 < t \leq |L|$, or $\pi/6 \leq \theta \leq \pi/2$, we actually have $P_m \geq 1/4$ for every integer $m > 0$. More precisely, we have the following lemma. 

\begin{lemma}
\label{lem:Pm_large_t}
For $|L|/4 < t \leq |L|$, which corresponds to $\pi/6 \leq \theta \leq \pi/2$, we have that
\[
    P_m \geq 
    \begin{cases}
        \frac{1}{4} & \text{for} \quad m=1 \\
        \frac{1}{2} - \frac{1}{1 - \frac{\pi^2}{96}} \approx 0.323 & \text{for} \quad m>1
    \end{cases}
\]
\end{lemma}

\begin{proof}
Let us start with the easy case of $m=1$. We have
\[
    \frac{\sin(4\theta)}{4\sin(2 \theta)} = \frac{1}{2}\cos(2\theta)\,
\]
which is upper bounded by $\frac{1}{4}$ for $\pi/6 \leq \theta \leq \pi/2$. Therefore\footnote{This makes sense, since $m=1$ corresponds to measuring directly without doing any Grover iterations, and therefore the probability of finding a marked item is at least $\frac{1}{4}$ when $t \geq |L|/4$.}, $P_1 \geq \frac{1}{4}$.

Next, we focus on the case $m>1$. We will prove that
\[
    \frac{\sin(4m\theta)}{4m\sin(2 \theta)} \leq  \frac{1}{2\pi}  \frac{1}{1 - \frac{\pi^2}{96}} \approx 0.177 \, .
\]
To do so, define $f(\theta) = \frac{\sin(4m\theta)}{4m\sin(2 \theta)} = \frac{g(\theta)}{h(\theta)}$, where $g(\theta) = \sin(4m\theta)$ and $h(\theta)=4m\sin(2 \theta)$. We have that $|g(\theta)| \leq 1$. Also, note that $f(\theta)$ is only positive when $g(\theta)\geq 0$. Furthermore, by L' H\^{o}pital's rule,
\[
    \lim_{\theta \rightarrow \frac{\pi}{2}} \frac{\sin(4m\theta)}{4m \sin(2\theta)} 
    = \lim_{\theta \rightarrow \frac{\pi}{2}} \frac{4m\cos(4m\theta)}{8m \cos(2\theta)} 
    = -\frac{1}{2}\, .
\]  
We now  distinguish two sub-cases: \\

\noindent\textit{Case (a): $ \pi/4 \leq  \theta \leq \pi/2$.}  On this interval for $\theta$ we have that $h(\theta)$ is a decreasing positive function. The roots of $g(\theta)$ are given by $\theta = \frac{\pi n}{4 m}$, $n\in \mathbb{Z}$. In particular, there is one root at $\pi/2$, which corresponds to  $n=2m$. Since $g(\theta)$ is negative just left of $\theta = \frac{\pi}{2}$, the largest value of $\theta$ for which $g(\theta)$ is positive is attained when $n \rightarrow 2m-1$, which corresponds to $\theta^*=\frac{\pi(2m-1)}{4m}$. Combining the above, we have that 
\begin{align*}
    f(\theta) \leq \frac{1}{4 m \sin (2\theta)} \leq \frac{1}{4 m \sin (2 \frac{
\pi(2m-1)}{4m})} = \frac{1}{4 m \sin (\frac{\pi}{2 m})}.
\end{align*}
We now use that $\sin(x) \geq x-\frac{1}{6}x^3$ for $x\geq 0$, which holds since
\begin{enumerate}
    \item  $\sin(x) = x-\frac{1}{6} x^3$ at $x=0$. 
    \item $\frac{\partial}{\partial x} \sin(x) = \cos(x) \geq 1-\frac{1}{2}x^2 = \frac{\partial}{\partial x} \left[ x-\frac{1}{6} x^3 \right]$, since
    \begin{enumerate}
        \item  $\cos(x) = 1-\frac{1}{2} x^2$ at $x=0$. 
        \item $\frac{\partial}{\partial x} \cos(x) = - \sin(x) \geq -x =  \frac{\partial}{\partial x}\left(1 - \frac{1}{2}x^2\right)$ for $x \geq 0$\\
        (using $\sin(x) \leq x$ for $x \geq 0$).
    \end{enumerate}
\end{enumerate}
For $m>1$, we have that $\tilde{m} := \frac{\pi}{2m} \in (0, \frac{\pi}{4}]$, and therefore
\begin{align*}
  \frac{1}{4 m \sin (\frac{\pi}{2 m})} 
  &\leq \frac{1}{4 m \left( \tilde{m}-\frac{1}{6}\tilde{m}^3\right)} = \frac{1}{4 m}\frac{1}{\tilde{m}} \frac{1}{1-\frac{1}{6}\tilde{m}^2} 
  \leq \frac{1}{4m} \frac{2m}{\pi} \frac{1}{1 - \frac{\pi^2}{96}} \\
  &= \frac{1}{2\pi}  \frac{1}{1 - \frac{\pi^2}{96}} \, .
\end{align*}

\noindent\textit{Case (b): $ \pi/6 \leq  \theta \leq \pi/4$.} On this interval $h(\theta)$ is an increasing positive function with minimum $\frac{1}{2 \sqrt{3}m}$ attained at $\theta=\pi/6$. For $m>1$, we have that 
\[
    f(\theta) \leq \frac{1}{2 \sqrt{3}m} \leq \frac{\sqrt{3}}{12} < \frac{1}{2\pi}  \frac{1}{1 - \frac{\pi^2}{96}}.
\]

\noindent Finally, since $P_m = 1/2-f(\theta)$, we conclude that $P_m \geq \frac{1}{2} -  \frac{1}{2\pi}  \frac{1}{1 - \frac{\pi^2}{96}}$ for $\pi/6 \leq \theta \leq \pi/2$.

% \noindent\textit{Case 3: $ \pi/6 \leq  \theta \leq \pi/2$, $m=1$.}
% Now all that is left to check is the case of $m=1$, which can be easily solved analytically:
% \begin{align*}
%     \min_{\theta \in [\pi/6,\pi/2]} f (\theta) = 1/4.
% \end{align*}

\end{proof}

 To start with, let us bound the expected number of queries of $\qsb$ to the quantum oracle.

\begin{lemma}\label{lem:E_Boyer}
The expected number of queries $E_{\qsb}^{\text{Quantum}}$ to the quantum oracle $\mO_g$ used by $\qsb$ when applied to a list $L$ with $t$ marked items can be upper bounded by
\begin{equation}
    E_{\qsb}^{\text{Quantum}}(|L|,t) \leq 
    \begin{cases}
        \frac{9}{2}m_t + \ceil{\log_{\lambda}(m_t)} - 3 &\text{if} \quad 1 \leq t < \frac{|L|}{4} \\
        2.0344  &\text{if} \quad \frac{|L|}{4} \leq t \leq |L| 
        \label{eq:ubound_queries_boyer}
    \end{cases}
\end{equation}
where $\lambda = \frac{6}{5}$. 
\end{lemma}

\begin{proof}
For every Grover cycle of $\qsb$, we sample $0 \leq j < m$ uniformly at random. Hence, on average, $j = (\ceil{m}-1)/2$. Because we need one query at the end of each cycle to check if the returned item is marked\footnote{In terms of queries to the function $g$, a check actually requires one query to $g$, not $c_q$ queries to $g$. Since $c_q \geq 1$, and we are working with upper bounds, and the checking only happens $\ceil{\log(m_t)}$ times, we choose to count the check as $c_q$ queries to $g$ to keep the formulas clean. As side note, when we say that we know the value of $R$ for $\textbf{QMax}$ or $\varphi$ for \maxsat `from the previous step', this value comes from these classical checks (that occur as a subroutine of both \textbf{QSearch} and \textbf{QMax}).}, on average a single Grover cycle uses $(\ceil{m}-1)/2 + 1 \leq m/2+1$ queries. Since $m$ is initialized at $m=\lambda$, and multiplied by $\lambda$ after every run, $m(u) = \lambda^{u}$ during the $u$-th cycle. 

Now, for $|L|/4 \leq t \leq |L|$ by\footnote{Lemma~\ref{lem:Pm_large_t} holds for integer $m$. In our case, $m(u) = \lambda^u$ is not necessarily integer, but since we are picking the integer $j$ less than $m(u)$, we can always replace by $m(u)$ by $\ceil{m(u)}$. In particular, since $m(1) = \lambda > 1$, $\ceil{m(u)} > 1$ for all $u$, and we can use the bound in Lemma~\ref{lem:Pm_large_t} for $m>1$.} Lemma~\ref{lem:Pm_large_t}, writing $p = \frac{1}{2} - \frac{1}{2\pi}\frac{1}{1- \frac{\pi^2}{96}}$, $\qsb$ has success probability lower bounded by $p$ and failure probability upper bounded by $\leq 1- p$. Hence, we can upper bound the expected number of queries for $\qsb$ by
\begin{align*}
    \sum_{u=1}^\infty \left(1-p\right)^{u-1} p\left(\frac{\lambda^{u}}{2} + 1\right) 
    &= \frac{\lambda p}{2} \sum_{u=0}^\infty \left(\lambda(1-p)\right)^{u} + p\sum_{u=0}^\infty \left(1-p\right)^{u} 
    = \frac{\lambda p}{2} \frac{1}{1-\lambda(1-p)} + 1 \\
    &= \frac{3}{5}\left(\frac{1}{2} - \frac{1}{2\pi}\frac{1}{1- \frac{\pi^2}{96}} \right)\frac{1}{1 - \frac{6}{5} \left( \frac{1}{2} + \frac{1}{2\pi}\frac{1}{1- \frac{\pi^2}{96}}\right)} +1 \leq 2.0344 \, .
\end{align*}

For $1 \leq t < \frac{|L|}{4}$, we can repeat the analysis of Boyer et al. Before we reach the \emph{critical stage}, i.e. the cycles for which $m < m_t$, the number of queries is upper bounded by 
\begin{align*}
    \sum_{u=1}^{\ceil{\log_{\lambda}(m_t)} - 1} \left(\frac{\lambda^{u}}{2} + 1\right) 
    &= \frac{1}{2}\frac{\lambda^{\ceil{\log_{\lambda}(m_t)}} -\lambda}{\lambda -1} + \ceil{\log_{\lambda}(m_t)} - 1 \\
    &\leq \frac{1}{2}\frac{\lambda m_t -\lambda}{\lambda -1} + \ceil{\log_{\lambda}(m_t)} - 1 \\
    &= 3m_t + \ceil{\log_{\lambda}(m_t)} - 4\, .
\end{align*}
Once $m \geq m_t$, we are in the critical stage, and by Lemma~\ref{lem:Pm_Boyer}, $P_m \geq \frac{1}{4}$. Hence, the expected number of queries is upper bounded by
\[
    \sum_{u=0}^\infty \left(\frac{3}{4}\right)^{u} \frac{1}{4} \left(\frac{\lambda^{u + \ceil{\log_{\lambda}(m_t)}}}{2} + 1 \right) 
    = \frac{\lambda^{\ceil{\log_{\lambda}(m_t)}}}{8} \frac{1}{1-\frac{3\lambda}{4}} + 1 \leq \frac{\lambda m_t}{8} \frac{1}{1-\frac{3\lambda}{4}} + 1  = \frac{3}{2}m_t + 1   \, .
\]
Including the upper bound to the expected number of queries before the critical stage, we arrive\footnote{The reason that we pick $\lambda = 6/5$ is that it minimizes the coefficient of the dominant term $m_t$ -- which by the above two expressions is given by $c(\lambda) = \frac{1}{2}\left(\frac{\lambda}{\lambda -1} + \frac{\lambda}{4 - 3\lambda}\right)$ -- on the interval $\lambda \in (1, 4/3)$. In particular, the choice $\lambda = 6/5$ is optimal.} at Eq.~\eqref{eq:ubound_queries_boyer}.

\end{proof}

It should be noted that the bound of $\frac{9}{2}m_t + \ceil{\log_{\lambda}(m_t)} - 3$ actually holds for all $t$, but only becomes useful when we can further bound $m_t$ (as we do in the next section, which requires that number of marked items is not too large, e.g. $1 \leq t \leq |L|/4$).

\subsection{Success probability}
\label{app:qsearch_succes_prob}

% Suppose that we have $t$ out of $|L|$ good vertices. The original paper of Boyer et al.~\cite{boyer1998tight} starts with drawing $c$ random samples classically to conclude, with high probability that $t\leq \frac{(a-1)|L|}{a}$. The latter bound we can then use the bound the failure probability of step 2 of the algorithm.

% However, it turns out that for step 2 to succeed with constant probability, we do not need that $t\leq \frac{(a-1)|L|}{a}$. Regardless, we will keep the step 1 as part of our implementation of \textbf{QSearch} because it can lead to a lower number of calls to the oracle.

% \subsubsection{Classical sampling}

% \ido{This part is not necessary anymore}

% We start with analysing the classical sampling part (step 1) of \textbf{QSearch}. . First, we need to check if $t > \frac{(a-1)|L|}{a}$ by drawing $c$ random samples classically. If we find a marked vertex, then we are done, and in particular we don't use the quantum routine in step 2. If we don't find a marked vertex we conclude that $t \leq \frac{(a-1)|L|}{a}$ with high probability, and proceed to step 2.

% The probability of this procedure failing: i.e $t > \frac{(a-1)|L|}{a}$ but we somehow managed to draw $c$ unmarked vertices is upper bounded by $\frac{1}{a^c}$, which means the success probability of step 1 is lower bounded by
% \[
%     p^{\text{success}}_{\text{classical}} \geq 1 - \frac{1}{a^c}\, .
% \]

% \subsubsection{Grover iterations}

Next, we focus on $\textbf{QSearch}$ as described by Algorithm~\ref{alg:QSearch}, which includes a time-out, making it a finite-time bounded-error algorithm whose success probability and complexity we analyse in the subsections that follow. 

As a consequence of Lemma~\ref{lem:Pm_large_t}, we do not need to first sample classically to exclude the case of many marked items before using Grover, because the success probability of a single run is $\geq \frac{1}{4}$ also in the regime of many marked items. Hence, the success probability of \textbf{QSearch} only depends on the success probability of the Grover search part (lines 6 - 19 of Algorithm~\ref{alg:QSearch}), which we investigate next. Note that, in our implementation of \textbf{QSearch} given in Algorithm~\ref{alg:QSearch}, we \emph{do} include the classical sampling part because it can make the algorithm efficient in the regime of many marked items.

If $t=0$, then the Grover search part will run the maximum number $N_{\text{runs}}$ of Grover runs, and return `no marked item found', which means it will always return the correct answer. Thus, we can restrict ourselves to the case $t>0$. In this case, the Grover part can only fail when every Grover run fails. For a single Grover run to fail, it has to not find a marked item before the time-out, meaning that $\qsb$ would have required more than $Q_{\text{max}}$ queries to find a marked item. 

\

Let $X$ be the random variable that corresponds to the number of queries to $\mO_g$ needed for $\qsb$ to find a marked item. We distinguish two cases.\footnote{We assume that $|L|\geq4$ throughout this analysis.}

\paragraph{Case 1: $1 \leq t\leq \frac{|L|}{4}$.} In this case, $\frac{|L|}{|L|-t} \leq \frac{4}{3}$, and therefore\footnote{Note that the bound $\frac{9}{4}\frac{|L|}{\sqrt{(|L|-t)t}} \leq \frac{9}{2} \sqrt{\frac{|L|}{t}}$ in~\cite{boyer1998tight} follows similarly.}
\begin{equation}
    m_t = \frac{1}{2} \frac{|L|}{\sqrt{(|L|-t)t}} = \frac{1}{2} \sqrt{\frac{|L|}{|L| -t}} \sqrt{\frac{|L|}{t}} \leq \sqrt{\frac{|L|}{3t}}\, .
    \label{eq:ubound_mt}
\end{equation}
Now, let us set $\Qm = \alpha \sqrt{|L|}$, where $\alpha$ will be determined below. Using Eq.~\eqref{eq:ubound_queries_boyer}, the probability that $X \geq \Qm$ can be upper bounded by Markov's inequality:
\begin{equation}
    \Pr[X \geq \Qm] \leq \frac{\E[X]}{\Qm} \leq \frac{\frac{9}{2} m_t + \ceil{\log_{\lambda}(m_t)} - 3}{\alpha \sqrt{|L|}} 
    \leq \frac{1}{{\alpha}}\left( \frac{3\sqrt{3}}{2\sqrt{t}} + \frac{\log_{\lambda}\bigg(\sqrt{\frac{|L|}{3t}}\bigg) - 2}{\sqrt{|L|}}\right)\, .
    \label{eq:fail_prob_small_t_accurate}
\end{equation}
We want to make this expression less than or equal to $\frac{1}{3}$ for all $t \geq 1$ and for all $|L|$, which we can accomplished by maximising the above expression with respect to $t$ (which sets $t=1$ in the above expression), and then choosing $\alpha$ to be
\begin{equation}
    \alpha \geq  \max_{x \geq 1} \left(\frac{9\sqrt{3}}{2} + 3\frac{ \log_{\lambda}\left(\frac{\sqrt{x}}{3} \right) - 2}{\sqrt{x}} \right) 
    = \frac{9\sqrt{3}}{2} + \frac{25}{36e \ln(\lambda)} \approx 9.1954 \, .
    \label{eq:alpha_max}
\end{equation}
To keep the notation simple, let us set
\begin{equation}
    \alpha = 9.2 \, .
    \label{eq:alpha}
\end{equation}
In particular, setting $\alpha = 9.2$ actually guarantees that
\begin{equation}
    \Pr[X \geq \Qm] \leq  \frac{1}{3\sqrt{t}}  \, ,
    \label{eq:fail_prob_small_t}
\end{equation}
for all $1 \leq t \leq \frac{|L|}{4}$. Indeed, making the rightmost expression in Eq.~\eqref{eq:fail_prob_small_t_accurate} less than or equal to $\frac{1}{3\sqrt{t}}$ is equivalent to
\begin{equation}
    \alpha \geq \left(\frac{9\sqrt{3}}{2} + 3\frac{\left\lceil \log_{\lambda}\left(\frac{\sqrt{x}}{3} \right) \right\rceil - 3}{\sqrt{x}} \right) \,
\end{equation}
for $x = \frac{|L|}{t}$, which holds because of Eq.~\eqref{eq:alpha_max}.

% Since $\ln(x) \leq \frac{x}{e}$, $\log_{\lambda}(x) \leq \frac{x}{\ln(\lambda)e}$, by Eq.~\eqref{eq:ubound_mt}, we get \chris{This bound is quite loose, does that matter?}

\paragraph{Case 2: $\frac{|L|}{4} < t \leq |L|$.} By Eq.~\eqref{eq:ubound_queries_boyer},
\begin{equation}
    \Pr[X \geq \Qm ] \leq \frac{2.0344}{\alpha \sqrt{|L|}} \, .
    \label{eq:fail_prob_large_t}
\end{equation}
which is also less than or equal to $\frac{1}{3}$ for $\alpha = 9.2$.

\paragraph{In conclusion}
Given failure probability of at most $\epsilon > 0$, recall that we execute at most $N_{\text{runs}} = \log_3(1/\epsilon)$ Grover runs. Therefore, the probability that \textbf{QSearch} succeeds satisfies
\[
    p^{\text{success}}_{\textbf{QSearch}} \geq \left(1 - \frac{1}{3^{N_{\text{runs}}}}\right) = 1 - \epsilon\, 
\]
as required.

\subsection{Expected number of queries}
\label{app:qsearch_exp_num_q}

In this section of the appendix, we upper bound the expected number of queries to $g$ made by \textbf{QSearch}.

\subsubsection{Classical sampling part}
For fixed $|L|$ and $t$, the probability that a vertex drawn uniformly at random is marked is given by the fraction $f = t/|L|$. Now, if we draw at most $N_{\text{samples}}$ classical samples uniformly at random, and then use Grover search if all $N_{\text{samples}}$ samples turn out to be unmarked, this takes a total of
\begin{align*}
    E_{\textbf{QSearch}} 
    % &= f \cdot 1 + (1-f)f \cdot 2 +  \ldots + (1-f)^{N_{\text{samples}}-1} f N_{\text{samples}} + (1-f)^{N_{\text{samples}}} (N_{\text{samples}}  E_{\text{Grover}}) \\
    &= \sum_{i=1}^{N_{\text{samples}}} f(1-f)^{i-1}i + (1-f)^{N_{\text{samples}}} (N_{\text{samples}} + c_q E_{\text{Grover}})\,. \numberthis \label{eq:N_gtot_series}
\end{align*}
queries to $g$ in expectation, where $E_{\text{Grover}}$ is the expected number of queries to the quantum oracle $\mO_g$ made by all $N_{\text{runs}}$ Grover runs of \textbf{QSearch} combined.

If $t = 0$, then the above expression becomes
\begin{equation}
    E_{\textbf{QSearch}} = N_{\text{samples}} + c_q E_{\text{Grover}}\, .
    \label{eq:tot_num_queries_t0}
\end{equation}
If $1<t\leq |L|$, then we can evaluate the geometric series. We have the following expression for the number of queries made by the classical sampling part
\[
    E(f) = \sum_{i=1}^{N_{\text{samples}}} f(1-f)^{i-1}i + (1-f)^{N_{\text{samples}}} N_{\text{samples}} \, 
\]
where $f = t/|L|$ is the fraction of marked items. The sum above can be evaluated as follows:
\begin{align*}
    f \sum_{i=1}^{N_{\text{samples}}} (1-f)^{i-1}i &= -f \frac{\partial}{\partial f} \sum_{i=1}^{N_{\text{samples}}} (1-f)^i 
    = f\frac{\partial}{\partial f} \left( \frac{(1-f)^{N_{\text{samples}}+1} - 1 + f}{f} \right) \\
    &= -(N_{\text{samples}}+1)(1-f)^{N_{\text{samples}}} - \frac{(1-f)^{N_{\text{samples}}+1}}{f} + \frac{1}{f}.
\end{align*}
Adding $(1-f)^{N_{\text{samples}}} N_{\text{samples}}$ gives that $E(f)$ can be rewritten as follows:
\begin{align*}
    E(f) &=-(1-f)^{N_{\text{samples}}} - \frac{(1-f)^{N_{\text{samples}}+1}}{f} + \frac{1}{f}   \\
    &= (1-f)^{N_{\text{samples}}}\left(-1 - \frac{1-f}{f} \right) + \frac{1}{f}   \\
    &= \frac{1}{f}\left(1 - (1-f)^{N_{\text{samples}}} \right) \numberthis \label{eq:Ef_series}\, .
\end{align*}

Hence, for $1 \leq t \leq |L|$, we conclude that\footnote{
Interestingly, this expression is the same as the one we would have obtained using the following procedure: with probability $(1-f)^{N_{\text{samples}}}$, we do Grover search, which takes $N^{\text{Grover}}_g$ steps, and with probability $1 - (1-f)^{N_{\text{samples}}}$, we classically sample vertices, of which we need $\frac{1}{f}$. The nice thing about the implementation we use, is that our implementation of \textbf{QSearch} mimics this behaviour without knowing $f$ in advance (meaning that we're not flipping a coin that returns heads with probability $(1-f)^{N_{\text{samples}}}$).}
\begin{equation}
    E_{\textbf{QSearch}} = \frac{1}{f}( 1 - (1-f)^{N_{\text{samples}}} ) + (1-f)^{N_{\text{samples}}} c_q E_{\text{Grover}} \, .
    \label{eq:tot_num_queries}
\end{equation}

\

\subsubsection{Grover part} Next, we investigate the expected number of queries $E_{\text{Grover}}$ to $\mO_g$ in the Grover part of $\textbf{QSearch}$.

\paragraph{No marked items in the list}
If $t=0$, then every run runs to completion without finding a marked element. In total, a single run executes at most $\Qm$ queries to $\mO_g$. Since for $t=0$ we perform $N_{\text{runs}}$, the expected total number of queries to $\mO_g$ in case of no marked items is upper bounded by
\begin{equation}
    E_{\text{Grover}} \leq N_{\text{runs}} \alpha  \sqrt{|L|} \leq 9.2 N_{\text{runs}} \sqrt{|L|} \, ,
    \label{eq:num_Grover_queries_t0} 
\end{equation}
where in the last inequality we used the expression for $\alpha$ in Eq.~\eqref{eq:alpha}. 

\paragraph{Marked items in the list}
i.e. $t > 0$. To start with, we want to bound the number of Grover iterations executed in a single run. To do so, we first examine a single run of $\qsb$, i.e.~without a timeout. For $k \in \mathbb{N}$, let us write $p_k$ for the probability that the random variable $X$ -- introduced in Appendix~\ref{app:qsearch_succes_prob} corresponding to the number of queries to $\mO_g$ needed for $\qsb$ to find a marked item -- assumes the value $k$, i.e.~the probability that a single run of $\qsb$ would have found a marked item using a total of $k$ queries. Then, in terms of the probabilities $\{p_k\}_{k\in \mathbb{N}}$, we have
\[
    \E[X] = \sum_{k=1}^\infty k \, p_k \, . 
\]

% In order to include the effects of the timeout, let us write $q_k$ for the probability that $\qsb$ finds a marked item using a total of $k$ queries conditioned on the premise that it stays within the time-out. In particular, this means that
% \begin{align*}
%     q_k = p_k \quad &\text{if} \quad 0 \leq k \leq \Qm \\
%     q_k < p_k \quad &\text{if} \quad \Qm < k < \Qm + \sqrt{|L|} \\
%     q_k = 0 \quad &\text{if} \quad k \geq \Qm + \sqrt{|L|}
% \end{align*}
% since every run of $\qsb$ that finds an item using at most $\Qm$ queries is not affected by the timeout, some runs of $\qsb$ that require between $\Qm$ and $\Qm+\sqrt{|L|}$ queries are affected by the timeout\footnote{Note that for every $k$ such that $\Qm < k \leq \Qm + \sqrt{|L|}$ there will be some runs that are affected by the time out: e.g. even for $k = \floor{\Qm} + 1$, the run that first does exactly $\Qm$ queries, then draws $j=0$ and coincidentally finds a marked item after measuring will contribute to $p_{\floor{\Qm} + 1}$ but not to $q_{\floor{\Qm} + 1}$. Since a similar situation happens for the other $k$'s such that $\Qm < k \leq \Qm + \sqrt{|L|}$, we have that $q_k$ is strictly less than $p_k$ in this interval.}, and every run of $\qsb$ that requires at least $\Qm + \sqrt{|L|}$ queries will be stopped prematurely. \\

A single Grover run fails exactly when $\qsb$ would have timed-out. Hence, the probability of a single Grover run succeeding is given by
\[
    q_{\s} = \sum_{k=0}^{\ceil{\Qm}-1} p_k \, .
\]
Conditioned on the outcome that step 2 finds a marked item -- which happens with probability $q_{\text{success}}$, a single run requires in expectation 
\begin{equation}
    Q_{\s} = \sum_{k=0}^{\ceil{\Qm}-1} k \, \frac{q_k}{q_\s}\, 
    \label{eq:Q_success}
\end{equation}
queries. Due to Lemma~\ref{lem:bound_exp_val_promise}, we have that
\begin{equation}
    Q_\s \leq \E[X].
    \label{eq:bound_num_grover_iter}
\end{equation}

\begin{lemma} \label{lem:bound_exp_val_promise}
Let $\mathbb{P} = \{p_k\}_{k=0}^\infty$ be a discrete probability on $\mathbb{N}$ and let $ E = \sum_{k = 0}^{\infty} p_k k$ be the expectation value of sampling a number from $\mathbb{N}$ according to $\mathbb{P}$. If there is a promise that the sampled number $k$ is less than some value $K \in \mathbb{N}$, the resulting probability distribution $\mathbb{P'}$ is renormalized by $p_K = \sum_{k=0}^K p_k$, that is: $\mathbb{P'} =  \left\{\frac{p_k}{p_K}\right\}_{k=0}^K$. Now, we claim that 
\[
    E' = \sum_{k=0}^K \frac{p_k}{p_K}  k  \leq E,
\]
i.e. the expectation value of drawing a number bellow $K$ according to $\mathbb{P}'$ is bounded from above by the original expectation value $E$.
\end{lemma}

\begin{proof}
We distinguish two cases. 

\noindent \emph{Case 1:} $K \leq E$. In this case, we have
\[
    E' = \sum_{k = 0}^K  \frac{p_k}{p_K}  i  \leq \sum_{k = 0}^K  \frac{p_k}{p_K}  K = K \leq E.
\]

\noindent \emph{Case 2:} $K > E$. Now we have
\[
    E = E'p_K + (1- p_K) \sum_{k = K + 1}^\infty \frac{p_k}{1-p_K} i \geq   E'p_K + (1- p_K) K > E'p_K + (1- p_K) E  \, ,
\]
which implies
\[
    E \geq E' \, .
\]
\end{proof}

Similarly, conditioned on the outcome that we do not find a marked item (due to the timeout) -- which happens with probability $q_{\text{fail}} = 1 - q_{\text{success}}$, the number of queries $Q_{\f}$ in a single run can trivially be bounded by
\begin{equation}
    Q_{\f} < \Qm = \alpha \sqrt{|L|} = 9.2 \sqrt{|L|}\, .
    \label{eq:ubound_Qfail}
\end{equation}

\

\noindent Because we execute at most $N_{\text{runs}}$ Grover runs, the expected number of queries in its entirety is given by the following sum 
\begin{align*}
    E_{\text{Grover}} &= \sum_{j=0}^{N_{\text{runs}}-1} q_\f^j (1-q_\f)(jQ_\f + Q_\s) + q_\f^{N_{\text{runs}}} N_{\text{runs}} Q_\f\,  \\
    &= Q_\s (1-q_\f)\sum_{j=0}^{N_{\text{runs}}-1} q_\f^j + Q_\f\left(q_\f^{N_{\text{runs}}} N_{\text{runs}} + (1-q_\f)\sum_{j=1}^{N_{\text{runs}}-1} q_\f^j j \right)
\end{align*}  
We can simplify the series above as follows
\begin{align*}
    q_\f^{N_{\text{runs}}} N_{\text{runs}} + (1-q_\f) \sum_{j=1}^{N_{\text{runs}}-1} q_\f^j \, j &=  q_\f^{N_{\text{runs}}} + q_\f\left(q_\f^{N_{\text{runs}}-1} (N_{\text{runs}}-1) + (1-q_\f) \sum_{j=1}^{N_{\text{runs}}-1} q_\f^{j-1} \, j\right) \\
    &= q_\f^{N_{\text{runs}}} + q_\f \left(\frac{1}{1-q_\f} (1 - q_\f^{N_{\text{runs}}-1}) \right) \\
    &= \frac{q_\f}{1-q_\f}\left(q_\f^{N_{\text{runs}}-1}(1-q_\f) + 1 - q_\f^{N_{\text{runs}}-1}\right) \\
    &= \frac{q_\f}{1-q_\f}(1-q_\f^{N_{\text{runs}}})
\end{align*}
where, going from the first to the second line we have used the derived expression for $E(f)$ in Eq.~\eqref{eq:Ef_series} with $f = 1-q_\f$ and $N_{\text{samples}}=N_{\text{runs}}-1$. We conclude that, for $t>0$,
\begin{equation}
    E_{\text{Grover}} = Q_\s (1-q_\f^{N_{\text{runs}}})+ Q_\f q_\f \frac{1-q_\f^{N_{\text{runs}}}}{1-q_\f}\, .
\end{equation}

Recall that the probability $q_\f$ that a single Grover run fails is given by
\[
    q_\f = \sum_{k = \ceil{\Qm}}^\infty p_k\, ,
\]
which is bounded by Eq.~\eqref{eq:fail_prob_small_t} if $1\leq t < \frac{|L|}{4}$, and Eq.~\eqref{eq:fail_prob_large_t} if $\frac{|L|}{4} < t \leq |L|$ -- see Appendix~\ref{app:qsearch_succes_prob}. However, we don't have a lower bound on $q_\f$, and therefore the best we can do with the $1-q_\f^{N_{\text{runs}}}$ term is to upper bound it by 1, which is equivalent to taking the $N_{\text{runs}} \rightarrow \infty$ limit. As a consequence, for $t>0$, our upper bound for $E_{\text{Grover}}$ is independent of the number of Grover runs $N_{\text{runs}}$. The resulting upper bound for $E_{\text{Grover}}$ is given by
\begin{equation}
    E_{\text{Grover}} \leq Q_\s + q_\f \frac{Q_\f}{1-q_\f}\, ,
    \label{eq:num_Grover_queries}
\end{equation}
where $Q_\f$ is bounded by Eq.~\eqref{eq:ubound_Qfail}, and $Q_\s$ given by Eq.~\eqref{eq:Q_success} will be bounded in the subsection below.

\subsubsection{In conclusion}

Given a list $L$ with $t$ marked items,  a failure probability of $\epsilon$, and a maximum number of classical samples $N_{\text{samples}}$, $\textbf{QSearch}(L, N_{\text{samples}}, \epsilon)$ executes at most $N_{\text{runs}} = \ceil{\log_3(1/\epsilon)}$ Grover runs. 

\

\noindent \textbf{If} $\bm{t=0}$, then by Eqs.~\eqref{eq:tot_num_queries_t0} and~\eqref{eq:num_Grover_queries_t0}, the expected total number of queries to $g$ is bounded from above by
\[
    E_{\textbf{QSearch}} \leq N_{\text{samples}} + \alpha c_q \ceil{\log_3(1/\epsilon)}) \sqrt{|L|} \, .
\]

\

\noindent \textbf{If} $\bm{t > 0}$, then by Eq.~\eqref{eq:tot_num_queries} we have
\[
    E_{\textbf{QSearch}} = \frac{|L|}{t}\left( 1 - \left(1-\frac{t}{|L|}\right)^{ N_{\text{samples}}} \right) + \left(1-\frac{t}{|L|}\right)^{ N_{\text{samples}}} c_q E_{\text{Grover}} \,
\]
By Eq.~\eqref{eq:num_Grover_queries}, $E_{\text{Grover}}$ satisfies the bounds below
\begin{align*}
    E_{\text{Grover}} &\leq Q_\s + q_\f \frac{Q_\f}{1-q_\f}\\
    &\leq Q_\s + \frac{E[X]}{\alpha \sqrt{|L|}} \frac{Q_\f}{1 - \frac{E[X]}{\alpha \sqrt{|L|}}}  \, .
\end{align*}
In the second line we have used Eq~\eqref{eq:fail_prob_small_t_accurate}. Now, for $Q_\f$, we have the upper bound from Eq.~\eqref{eq:ubound_Qfail}: $Q_\f \leq \alpha\sqrt{|L|}$. Moreover, to bound $Q_\s$, we can use Eq.~\eqref{eq:bound_num_grover_iter}, which says that $Q_\s \leq E[X]$. Hence, we obtain
\[
    E_{\text{Grover}} 
    \leq E[X] \left(1 + \frac{1}{1 - \frac{E[X]}{\alpha \sqrt{|L|}}} \right) \, ,
\]
Depending on the number of marked items $t$, we can bound $E[X]$ as follows.
\begin{itemize}
    \item If $1 \leq t \leq \frac{|L|}{4}$, then  by Eq.~\eqref{eq:ubound_queries_boyer},
    \[
        E[X] \leq \frac{9}{4}\frac{|L|}{\sqrt{(|L| -t )t}} + \left\lceil\log_{\frac{6}{5}}\left(\frac{|L|}{2\sqrt{(|L| -t )t}} \right) \right\rceil - 3 \leq \frac{\alpha \sqrt{L|}}{3 \sqrt{t}}.
    \]
    The rightmost inequality follows from the analysis leading up to Eq.~\eqref{eq:fail_prob_small_t}.
    
    \item In case $\frac{|L|}{4} < t \leq |L|$, by Eq.~\eqref{eq:ubound_queries_boyer}  we have 
    \[
        E[X] \leq 2.0344 \, .
    \]
\end{itemize}
In the above formulas, by Eq.~\eqref{eq:alpha},
\[
    \alpha = 9.2 \geq \max_{x \geq 1} \left(\frac{9\sqrt{3}}{2} + 3\frac{\log_{\lambda}\left(\frac{\sqrt{x}}{3} \right)  - 2}{\sqrt{x}} \right) \,.
\]
The expressions obtained for the expected number of queries to $g$ are presented more concisely in Lemma~\ref{lem:QSearch}.

\subsection{Worst-case behaviour of \ZQ}\label{app:zalka_qsearch_proof}
The two steps of this algorithm are given in Section~\ref{sec:qsearch_worst_case}. Here we analyse the worst-case complexity of that implementation. For the first step, recall that (see Lemma~\ref{lem:exact_grover}) when there are $t$ marked items (and we know $t$), then exact Grover search\footnote{We note that exact Grover search is somewhat unrealistic, since it requires arbitrarily precise rotations to be performed in each Grover step. In general this will not be possible, and the best we can hope for is some very close approximation using a sequence of gates taken from some universal gate set. This approximation can be made inverse-exponentially close to the correct rotations with only a polynomial overhead (in terms of the number of gates required), which in particular does not contribute to the query complexity of the algorithm. This will mean that `exact' Grover search will actually fail with some small probability, but that this probability can be made small enough to be negligible for our purposes (e.g. we can simply ensure that the probability of it failing is smaller than $\epsilon/t_0$ without incurring any extra queries). With this in mind, we ignore such considerations, since there will always be issues of approximation and error arising from the physical implementations of quantum algorithms, and we view such issues as being on the same level as overheads from error correction, which as we have already discussed we will omit from our analysis.} can find and return one \textit{with certainty}, using precisely $\round*{\frac{\pi}{4}\sqrt{\frac{N}{t} - \frac12}} + 1$ Grover iterations~\cite{hoyer2000arbitrary}.\footnote{The notation $\lceil x \rfloor$ represents the closest integer to $x$.} Therefore, step 1 requires at most 
\begin{eqnarray*}
    \sum_{t=1}^{ t_0 } \left(\round*{\frac{\pi}{4}\sqrt{\frac{N}{t}} - \frac12} + 1\right) \leq \sum_{t=1}^{ t_0 } \left(\frac{\pi}{4} \sqrt{\frac{|L|}{t}} + 1\right) &=& t_0 + \frac{\pi}{4}\sqrt{|L|} \sum_{t=1}^{ t_0 } \frac{1}{\sqrt{t}} \\
    &\leq&  t_0 + \frac{\pi}{4}\sqrt{|L|} \int_{0}^{t_0} \frac{1}{\sqrt{t}} dt  \\
    &=& t_0 + \frac{\pi}{2}\sqrt{|L| t_0}\,
\end{eqnarray*}
queries to the quantum oracle $\mO_g$ (recalling that every Grover iteration takes a single query, and every Grover run requires one at the end to check whether a marked item was found or not). For the second step, clearly there will be at most
% \ido{I get this}
% \[
%     \left(2\frac{1}{2}\left( \frac{\pi}{4} \sqrt{\frac{|L|}{t_0}} + 1\right) +1 \right) t_0 = \frac{\pi}{4} \sqrt{|L|t_0} + 2 t_0 \ido{?}
% \]
\[
    2t_0\ceil*{\frac{\pi}{4}\sqrt{\frac{|L|}{t_0}}} + 2t_0 \leq \frac{\pi}{2} \sqrt{|L| t_0} + 4t_0
\]
queries (though the probability of doing this many is quite small), and it remains to set the value of $t_0$. Suppose that whenever we run Grover search with a number of iterations in the range $[0,\lceil{\frac{\pi}{4} \sqrt{\frac{|L|}{t_0}}}\rceil]$, the probability of detecting a marked item is at least $p$. Then the probability that this step of the algorithm fails will be no more than $(1-p)^{2t_0}$. If we require that $(1-p)^{2t_0} \leq \epsilon$, then it is sufficient to choose 
\[
    t_0 = \ceil*{ \frac{\ln\left(\epsilon\right)}{2 \ln \left(1-p\right)} }
\]
(where we round up to ensure that $t_0$ is an integer). Then we can combine the queries from each stage and rewrite the total number made by the algorithm as
\[
    5\ceil*{\frac{\ln \epsilon}{2\ln(1-p)}} + \pi \sqrt{|L|} \sqrt{\ceil*{\frac{\ln \epsilon}{2\ln(1-p)}}}\,.
\]
It remains to lower bound the success probability $p$ for different values of $t$. We can first use a bound derived in~\cite{boyer1998tight}, which says that $p \geq 1/4$ provided $\ceil*{\frac{\pi}{4} \sqrt{\frac{|L|}{t_0}}} \geq \frac{|L|}{2\sqrt{t}\sqrt{|L|-t}}$. Using the fact that $t > t_0$, and further assuming that $t \leq \frac{(a-1)|L|}{a}$ for some $a>1$ to be chosen, we have 
\[
    \frac{|L|}{2\sqrt{t}\sqrt{|L|-t}} \leq \frac{1}{2}\sqrt{\frac{a|L|}{t_0}}\,
\]
and therefore we can ensure that $\ceil*{\frac{\pi}{4} \sqrt{\frac{|L|}{t_0}}} \geq \frac{|L|}{2\sqrt{t}\sqrt{|L|-t}}$ by choosing $a = \frac{\pi^2}{4}$, in which case we obtain $p \geq 1/4$.

Now we can consider the case that $t > \frac{(a-1)|L|}{a}$. In this case, we are still very likely to find a marked item even after applying some Grover iterations that will, in general, rotate the state away from the marked subspace, and as we show in Lemma~\ref{lem:Pm_large_t} we have $p \geq 1/4$ when $t > \frac{\pi^2/4 - 1}{\pi^2/4} > \frac{|L|}{4}$, and so for both cases we have $p \geq 1/4$. Plugging this lower bound on $p$ into the expression for the total number of Grover iterations, and taking into account that a single query to $\mO_g$ corresponds to $c_q$ queries to $g$, we see that the total number of queries to $g$ made by the algorithm is at most
\begin{equation}
    % W_{\textbf{QSearch}_{\text{Zalka}}}(|L|,\epsilon) :=  6.953 \ln(1/\epsilon) + 8.284 \sqrt{|L|} \sqrt{\ln (1/\epsilon)}\,.
    W_{\textbf{QSearch}_{\text{Zalka}}}(|L|,\epsilon) := c_q \left(5\ceil*{\frac{\ln (1/\epsilon)}{2\ln(4/3)}} + \pi \sqrt{|L|} \sqrt{\ceil*{\frac{\ln (1/\epsilon)}{2\ln(4/3)}}}\right) \,.
\end{equation}

\section{Detailed analysis of \textbf{QMax}}
\label{app:QMax}

In this section we compute the expected number of queries to $g$ made by $\qmi$, and we give further upper bounds to the obtained expression for the expected number of queries.

\subsection{Expected number of queries}
\label{app:QMax-run-time}

Based on the the proof idea of~\cite{ahuja1999quantum}, we provide a more accurate proof and expression of the expected number of queries made by $\qmi$ when searching the list $L$ with $t$ marked items, as stated in Lemma~\ref{lem:EQMax-inf} (restated below for convenience). 

\lemmaEQMaxinf*

\begin{proof}
Let us use the shorthand notation
\[
    Q(t) = E_{\qsb}^{\text{Quantum}}(|L|, t)
\]
for the expected number of queries made by $\qsb$ to the quantum oracle when searching a list $L$ with $t$ marked items (suppressing the $L$ dependence for notational convenience). Note that by Lemma~\ref{lem:E_Boyer}, 
\[
    Q(t) \leq F(|L|,t)
\]
where $F(|L|,t)$ is given by Eq.~\eqref{eq:F(L,t)}. Moreover, let $E(t)$ denote the expected number of queries to the quantum oracles $\mO_{f_i}$ for finding the maximum when $t$ items are marked: i.e.~the expected number of queries to find the maximum given that $y$ is set to the $t+1$-th item of $L$ when ordered according to $R$ in descending order. We first compute the expected number of queries to the quantum oracles $\mO_{f_i}$, and then include the factor of $c_q$ in the end\footnote{Since each query to $\mO_g$ corresponds to $c_q$ queries to $g$.}.

We have the following recursion relation for $E(t)$:
\begin{equation}
    E(t) = \frac{1}{t}\big(E(t-1) + E(t-2) + \ldots + E(1) + E(0) \big) + Q(t) \,,
    \label{eq:E_recursion}
\end{equation}
(because, after applying $\qsb$, with equal probability we find one of the $t$ marked items in $L$ with a larger value for $R$ than the current index $y$). Note that $E(0) = 0$.

Using Eq.~\eqref{eq:E_recursion} for $t$ and $t-1$:
\begin{align*}
    tE(t) &= \sum_{u=1}^{t-1} E(u) + tQ(t) \\
    (t-1)E(t-1) &= \sum_{u=1}^{t-2} E(u) + (t-1)Q(t-1) \,
\end{align*}
and subtracting the bottom equation from the top equation and then dividing by $t$ yields
\begin{equation}
    E(t) = E(t-1) + Q(t) - \frac{t-1}{t}Q(t-1) \, .
    \label{eq:E_intermediate_result}
\end{equation}

Since the above equation holds for every $t$, we can use the equation for $t-1$ and plug it into Eq.~\eqref{eq:E_intermediate_result}, and then do the same for $t-2$, etc., up to $t=2$. We then obtain
\[
    E(t) = E(1) + \sum_{u=2}^t \left( Q(u) - \frac{u-1}{u} Q(u-1)\right) \, .
\]  
We next rewrite\footnote{This is where the proof of~\cite{ahuja1999quantum} becomes imprecise. The authors use upper bounds for $Q(u)$ and $Q(u-1)$ in order to upper bound $E(t)$. However, in order to obtain an upper bound for $E(t)$, a \emph{lower} bound for the $Q(u-1)$ terms should be used, because they come with a minus sign. The correct way to continue the proof is to postpone using upper bounds for $Q(u)$ until after rewriting the sum.} the sums above as follows:

\begin{align*}
    E(t) &= E(1) + \sum_{u=2}^t Q(t) - \sum_{u=1}^{t-1} \frac{u}{u+1} Q(u) \\
    &= E(1) + Q(t) + \sum_{u=2}^{t-1} Q(t)\left(1 - \frac{u}{u+1}\right) - \frac{1}{2}Q(1) \\
    &= Q(t) + \sum_{u=2}^{t-1} \frac{Q(u)}{u+1} + \frac{1}{2}Q(1) \\
    &= Q(t) + \sum_{u=1}^{t-1} \frac{Q(u)}{u+1} \, \numberthis \label{eq:E_series},
\end{align*}
where, when going to the third line, we have used that $E(1) = Q(1)$.

Now, since, at initialization, $\qmi$ chooses an index $y$ uniformly at random, the expected number of queries of $\qmi$ to the quantum oracles is given by
\[
    \frac{1}{|L|} \sum_{t=1}^{|L|-1} E(t) \,.
\]
As before, $E(0) = 0$ and needs not to be included in the sum. Using Eq.~\eqref{eq:E_series}, we get
\begin{align*}
    \frac{1}{|L|} \sum_{t=1}^{|L|-1} E(t) &= \frac{1}{|L|} \left(\sum_{t=1}^{|L|-1} Q(t) + \sum_{t=1}^{|L|-1} \sum_{u=1}^{t-1} \frac{Q(u)}{u+1} \right) \\
    &= \frac{1}{|L|} \sum_{s=1}^{|L|-1} Q(s)\left(1 + \frac{|L| - 1  -s}{s+1} \right) \\
    &= \sum_{s=1}^{|L|-1} \frac{Q(s)}{s+1} \, ,
\end{align*}
where the expression is the second line can be obtained by collecting all terms $Q(s)$ together for every $s \in \{1, 2, \ldots, |L|-1\}$. Including the factor $c_q$, and \emph{only now} using the fact that $Q(s) \leq F(|L|,s)$ we obtain our result.
\end{proof}

\subsection{Upper bounds to the expected number of queries}
\label{app:QMax-Ubound}

Next, we will upper bound the series that upper bounds the expected number of queries to $g$ made by $\qmi$ by a simpler expression. By Lem.~\ref{lem:EQMax-inf} and Eq.~\eqref{eq:F(L,t)}, we have
\begin{equation}
    E_{\qmi}(|L|) \leq c_q \left( \sum_{s=1}^{\ceil{|L|/4}-1} \frac{F(L,s)}{\sqrt{s}(s+1)} + 2.0344 \sum_{s=\ceil{|L|/4}}^{|L|-1} \frac{1}{s+1} \right)
    \, .
    \label{eq:EQMax_upper_bound_app}
\end{equation}
with
\[
    F(L,s) \leq \frac{9}{4}\frac{|L|}{\sqrt{(|L| -t )t}} + \left\lceil\log_{\frac{6}{5}}\left(\frac{|L|}{2\sqrt{(|L| -t )t}} \right) \right\rceil - 3 \leq \frac{9.2 \sqrt{L|}}{3 \sqrt{t}} 
\]
for $s \leq \ceil{|L|/4}-1$.

\paragraph{Loose upper bound}
We can upper bound the sums above by integrals. For the second term, we have
\[
    \sum_{s=\ceil{|L|/4}}^{|L|-1} \frac{1}{s+1} \leq \int_{\ceil{|L|/4}-1}^{|L|-1}ds\frac{1}{s+1} 
    = \int_{\ceil{|L|/4}}^{|L|}dx\frac{1}{x} = \ln(|L|) - \ln(\ceil{|L|/4}) \leq \ln(4) \, . 
\]
For the first term, using $\frac{9.2\sqrt{|L|}}{2\sqrt{s}}$ as an upper bound for $F(L,s)$, we get
\[
    \sum_{s=1}^{\ceil{|L|/4}-1} \frac{1}{\sqrt{s}(s+1)} \leq \frac{1}{2} + \int_{1}^{\ceil{|L|/4}-1} ds\,  \frac{1}{\sqrt{s}(s+1} = \frac{1}{2} + 2\arctan(\sqrt{\ceil{|L|/4}-1}) - \frac{\pi}{2}
\]
Hence,
\begin{align*}
    E_{\qmi}(|L|) &\leq c_q \left(\frac{9.2\sqrt{|L|}}{3}\left(\frac{1-\pi}{2} + 2\arctan(\sqrt{\ceil{|L|/4}-1}) \right) + 2.0344 \ln(4) \right) \\
    &\leq c_q \left( \frac{9.2\sqrt{|L|}}{3} \frac{1+\pi}{2}  + 2.0344 \ln(4) \right) \\
    & \leq c_q \left( 6.3505 \sqrt{|L|} + 2.8203 \right)\, .
\end{align*}

\paragraph{Tighter bound}

We can also use
\[
    F(L,t) \leq \frac{3 \sqrt{3}}{2} \sqrt{\frac{|L|}{t}} + \log_{\frac{6}{5}}\left(\sqrt{\frac{|L|}{3t}} \right) - 2 
\]
for the $1 \leq t < |L|/4$ regime -- obtained from the tighter upper bound for $F(L,t)$ above combined with Eq.~\eqref{eq:ubound_mt}. In this case, the first term in Eq.~\eqref{eq:EQMax_upper_bound_app} can be upper bounded by 
\begin{align*}
    \sum_{s=1}^{\ceil{|L|/4}-1} \frac{F(L,s)}{s+1}
    &\leq \frac{3\sqrt{3|L|}}{2} \sum_{s=1}^{\ceil{|L|/4} -1} \frac{1}{\sqrt{s}(s+1)} 
    + \left( \frac{1}{2}\log_{\frac{6}{5}}(|L|/3) -2\right) \sum_{s=1}^{\ceil{|L|/4} -1} \frac{1}{s+1} \\ 
    & \quad - \frac{1}{2}\sum_{s=1}^{\ceil{|L|/4} -1} \frac{\log_{\frac{6}{5}}(s)}{s+1} \, .
\end{align*}
The first term above can be bounded by
\[
    \frac{3\sqrt{3|L|}}{2} \left(\frac{1-\pi}{2} + 2\arctan\left(\sqrt{\ceil{|L|/4}-1}\right) \right) \, ,
\]
as before, and the second by
\[
    \left( \frac{1}{2}\log_{\frac{6}{5}}(|L|/3) -2\right) \ln(|L|/4) \, .
\]
For the third term above, note that the function $\frac{\ln(s)}{s+1}$ is monotonically decreasing for $s \geq 4$. Hence, assuming $\ceil{|L|/4}-1 \geq 4$, we have
\begin{align*}
    \frac{1}{2 \ln(6/5)}\sum_{s=1}^{\ceil{|L|/4}-1} \frac{\ln(s)}{s+1} 
    &\leq \frac{1}{2 \ln(6/5)}\sum_{s=1}^{4} \frac{\ln(s)}{s+1} + \int_4^{\ceil{|L|/4}-1} ds \frac{\ln(s)}{s+1} \\
    &\leq 2.5279 + \frac{\text{Li}_2(-\ceil{|L|/4}+1) + \ln(\ceil{|L|/4}-1)\ln(\ceil{|L|/4})}{2 \ln(6/5)}
\end{align*}
where $\text{Li}_2$ is Spence's function, or dilogarithm.

Collecting all terms together, we get
\begin{align*}
    \sum_{s=1}^{\ceil{|L|/4}-1} \frac{Q(s)}{s+1}
    &\leq  \frac{3\sqrt{3|L|}}{2} \left(\frac{1-\pi}{2} + 2\arctan\left(\sqrt{|L|/4}\right) \right) \\
    &\quad + \frac{\ln(|L|/4)}{2 \ln(6/5)}\left(\ln(|L|/3) + \ln(|L|/4 + 1)  \right) \\
    &\quad -2\ln(|L|/4) + 2.5279 + \frac{\text{Li}_2(-\ceil{|L|/4}+1)}{2 \ln(6/5)} \, .
\end{align*}
Using $\arctan(x) \leq \frac{\pi}{2}$, the expected number of queries to $g$ of $\qmi$ is bounded by
\begin{align*}
    E_{\qmi}(|L|) &\leq c_q \Bigg[ \frac{3 \sqrt{3}(1+\pi)}{4}\sqrt{|L|} + \frac{\ln(|L|/4)}{2 \ln(6/5)}\bigg(\ln(|L|/3) + \ln(|L|/4 + 1)  \bigg) -2\ln(|L|/4) \\
    &\quad + 5.3482 + \frac{\text{Li}_2(-\ceil{|L|/4}+1)}{2 \ln(6/5)} \Bigg] \numberthis \label{eq:EQMax_sharp_bound}
\end{align*}
which has a leading term that grows as
\[
    c_q \frac{3 \sqrt{3}(1+\pi)}{4}\sqrt{|L|} \leq c_q 5.3801 \sqrt{|L|}\, .
\]

\section{Estimators for the expected number of queries for 
\textbf{QSearch}}
\label{app:jensens}

In this section, we provide additional details for our derivation of the estimator in Section~\ref{sec:estimating_marked_items} that overestimates expected number of queries for \textbf{QSearch} by sampling items from $L$ uniformly at random and counting after how many samples we find a marked item. We start by proving Lemma~\ref{lem:sampling_constant} and Lemma~\ref{lem:sampling_constant_log}, then construct $E_{\text{Grover}}^{\text{estimator}}$, and finally construct an estimator for the expected number of queries $E_{\textbf{QSearch}}$ for \textbf{QSearch}.

\subsection{Proof of Lemma~\ref{lem:sampling_constant}}
\label{app:proof_square_root}

Here we prove Lemma~\ref{lem:sampling_constant}, restated below for convenience.
\jensens*
\begin{proof}
The expectation value of $\mathbb{E} [\sqrt{d_1 X}]$ is given by
\begin{equation*}
     \mathbb{E} [\sqrt{d_1 X}] = \sum_{i=1}^\infty \sqrt{d_1 i} f (1-f)^{i-1} = \sqrt{d_1} \frac{f \text{Li}_{-\frac{1}{2}}(1-f)}{1-f} =: \sqrt{d_1} h(f),
\end{equation*}
where $\text{Li}_{s}(z)$ is a polylogarithmic function known as Jonquière's function, given by
\begin{equation*}
    \text{Li}_{s}(z) = \sum_{k-1}^\infty \frac{z^k}{k^s}.
\end{equation*}

We prove this corollary in three steps: 1) we investigate the limiting behaviour of the ratio of $\sqrt{d_1} h(f)/\sqrt{1/f}$ at the domain boundaries, and show that there exists a constant $d_1=\frac{4}{\pi}$ such that at both limits the ratio is at least one; 2) we show that this ratio is non-decreasing on this domain; 3) we deduct the implied upper bound and relative errors.
We obtain the limiting behaviour when $f\rightarrow 0^+$ as follows:
\begin{align*}
 \lim_{f\rightarrow 0^+}  \frac{\sqrt{d_1} h(f)}{\sqrt{\frac{1}{f}}} & =  \lim_{f\rightarrow 0^+}  \frac{\sqrt{d_1} f \text{Li}_{-\frac{1}{2}}(1- f)}{(1-f) \sqrt{\frac{1}{f}} } &&\\
 & = \sqrt{d_1} \lim_{f\rightarrow 0^+}  \frac{ \text{Li}_{-\frac{1}{2}}(1-f)}{ f^{-3/2}(1-f)  } &&\\
 & =\sqrt{d_1} \lim_{f\rightarrow 0^+}  \left[ \frac{1}{1-f} \right] \lim_{f\rightarrow 0^+}  \left[ \frac{\text{ Li}_{-\frac{1}{2}}(1-f)}{f^{-3/2}}   \right] \qquad &&\text{(Product rule)} \\
  &=  \sqrt{d_1} \lim_{f\rightarrow 0^+}  \frac{ \frac{\sqrt{\pi}}{2} f^{-3/2} + \mO(f^{-1/2})}{f^{-3/2}}  \qquad &&\text{(Series expansion)}\\
 & = \sqrt{d_1} \frac{\sqrt{\pi}}{2}. &&
\end{align*}

Similarly, for $f\rightarrow 1$ we get
\begin{align*}
\lim_{f\rightarrow 1} \sqrt{d_1} \frac{h(f)}{\sqrt{\frac{1}{f}}} &= \lim_{f\rightarrow 1}  \sqrt{d_1} \frac{f \text{Li}_{-\frac{1}{2}}(1- f)}{(1-f) \sqrt{\frac{1}{f}} }  &&\\
 & = \sqrt{d_1} \lim_{f\rightarrow 1}  \frac{ \text{Li}_{-\frac{1}{2}}(1-f)}{ f^{-3/2}(1-f)  } &&\\
 & =\sqrt{d_1} \lim_{f\rightarrow 1}  \left[ \frac{1}{f^{-3/2}} \right] \lim_{f\rightarrow 1}  \left[ \frac{\text{ Li}_{-\frac{1}{2}}(1-f)}{  1-f }   \right]  \qquad &&\text{(Product rule)} \\
 & = \sqrt{d_1} \lim_{f\rightarrow 1}  \frac{ (1-f) + \mO((1-f)^{2})}{1-f}\qquad &&\text{(Series expansion) }\\
 & = \sqrt{d_1}.  &&
\end{align*}
Therefore, if we pick $d_1=(2/\sqrt{\pi})^2=4/\pi$, we have that the ratio is at both limits at least one. We will now establish the non-decreasing property of this ratio, already evaluated with our proposed $\sqrt{d_1}=2/\sqrt{\pi}$, by invoking the first derivative test. Note that
\begin{align*}
    \frac{\partial}{\partial f} \frac{2}{\sqrt{\pi}} \frac{h(f)}{\sqrt{\frac{1}{f}}} & =  - \frac{2 f \text{Li}_{-3/2}(1-f) + (f - 3) \text{Li}_{-1/2} (1-f )}{\sqrt{\pi} \sqrt{\frac{1}{f}}(1 - f)^2},
\end{align*}
for which the denominator $\sqrt{\pi} \sqrt{\frac{1}{f}}(1 - f)^2 \geq 0$ for all $f \in (0,1]$, and the numerator has exactly one root at $f =1$. Since for any other $\tilde{f}\in (0,1]$ taken left of $f=1$ the derivative of the ratio is larger than zero, we must have that the ratio is a non-decreasing function. Therefore, we have that $\sqrt{d_1} h(f)\geq\sqrt{1/f}$ for all $f \in (0,1]$. 

% Finally, we have to establish the relative errors $\sqrt{d_1}$ gives. Since $| \sqrt{d_1} h(f)- \sqrt{1/f}| = h(\sqrt{d_1} f)- \sqrt{1/f}$ for all $f \in (0,1]$ and the fact that the largest relative error is attained when $ c_2h(f)/ \sqrt{1/f}$ is maximized by the established non-decreasing property, the largest relative error is given by
% \begin{align*}
% \lim_{f\rightarrow 1} \frac{2}{\sqrt{\pi}} \frac{h(f)}{\sqrt{\frac{1}{f}}}-1 = \frac{2}{\sqrt{\pi}}-1 \approx 0.1284.
% \end{align*}
% Also, the error for small $f$ (in the limit of going to zero) the relative error becomes (by using the already established limit for the ratio when $f$ goes to zero):
% \begin{align*}
% \lim_{f\rightarrow 0^+}  \frac{2}{\sqrt{\pi}} \frac{ h(f)}{\sqrt{\frac{1}{f}}}-1 = 0.
% \end{align*}

\end{proof}

\subsection{Proof of Lemma~\ref{lem:sampling_constant_log}}
\label{app:proof_logarithm}

\jensenslog*

\begin{proof}
For random variable $X$ geometrically distributed with parameter $f$, we have that  $\mathbb{E}[X]=\frac{1}{f}$, and
\[
    \mathbb{E}[\log{X}] = \sum_{i=1}^\infty \log{(i)} f (1-f)^{i-1}.
\]
% Since $\log{(x)}$ is concave, we must have that 
% \begin{align}
%      \log({\mathbb{E}[X]}) \geq \mathbb{E} [\log{(X)}]
% \end{align}
% by Jensen's inequality. 
We will show that 
\begin{equation}
    \log({\mathbb{E}[X]}) - \mathbb{E} [\log{X}] \leq \gamma \, 
    \label{eq:log_to_prove}
\end{equation}
from which the statement of the lemma follows:
\[
    \log(\E[X]) \leq \log(e^{\gamma}) + \E[\log X] = \E[\log (e^\gamma X)]\, .
\]

In order to prove Eq.~\eqref{eq:log_to_prove}, let us define $q:=1-f \in [0,1)$. Now, we have
\begin{align*}
\log({\mathbb{E}[X]}) - \mathbb{E} [\log{X}] &= -\log{(1-q)} - \sum_{i=1}^{\infty} \log{(i)} (1-q)q^{i-1} &&\\
&= -\log{(1-q)} - \sum_{i=2}^{\infty} \log{(i)} q^{i-1} + \sum_{i=1}^{\infty} \log{(i)} q^{i} &&\\
&= -\log{(1-q)} - \sum_{i=1}^{\infty} \log{(i+1)} q^{i} + \sum_{i=1}^{\infty} \log{(i)} q^{i} &&\\
&= -\log{(1-q)} - \sum_{i=1}^{\infty} \log{\left(\frac{i+1}{i}\right)} q^{i} &&\\
&= \sum_{i=1}^{\infty} \frac{q^i}{i} - \sum_{i=1}^{\infty} \log{\left(\frac{i+1}{i}\right)} q^{i} &&\qquad[\text{Series expansion}]\\
&= \sum_{i=1}^{\infty} \left[\frac{1}{i}-\log{\left(\frac{i+1}{i}\right)} \right]q^{i} &&
\end{align*}

\noindent Next, we investigate the above series without the $q^i$'s, which turns out to be convergent. Indeed, for $n\in \bbN$ we have
\begin{align*}
\sum_{i=1}^{n} \left[-\log{\left(\frac{1+i}{i}\right)} +\frac{1}{i} \right]
&= -\log{\left(\prod_{i=1}^n \frac{1+i}{i}\right)} + \sum_{i=1}^{n} \frac{1}{i} \\
&= -\log{\left(n+1\right)} + \sum_{i=1}^{n+1} \frac{1}{i} - \frac{1}{n+1}\, .
\end{align*}
Taking the limit $n\rightarrow \infty$ yields
\[
    \sum_{i=1}^{\infty} \left[-\log{\left(\frac{1+i}{i}\right)} +\frac{1}{i} \right]
    = \lim_{n \rightarrow \infty}  \left[-\log{\left(n+1\right)} + \sum_{i=1}^{n+1} \frac{1}{i} \right] - \lim_{n \rightarrow \infty} \left[\frac{1}{n+1}\right] = \gamma \approx 0.5772
\]
by definition of the Euler-Mascheroni constant. 

Moreover, since 
\begin{align*}
    \frac{1}{i}-\log{\left(\frac{i+1}{i}\right)}  = \frac{1}{i}-\log{\left(1+\frac{1}{i}\right)} \geq 0 \text{ for all }i>0,
\end{align*}
and $q^i \in [0,1)$, the following must hold
\begin{align*}
    \sum_{i=1}^{\infty} \left[\frac{1}{i}-\log{\left(\frac{i+1}{i}\right)} \right]q^{i} \leq \gamma \text{ for all } 0\leq q < 1 \, ,
\end{align*}
and therefore so does Eq.~\eqref{eq:log_to_prove}.

\end{proof}

\subsection{Estimator for $E_{\text{Grover}}$ for all $t \geq 1$}
\label{app:additive_constant}

We next construct an estimator that overestimates $E_{\text{Grover}}$. To do so, we will first construct a single function of the form
\begin{equation}
    G(f) = c_0 + c_1 \frac{1}{f} + c_2 \sqrt{\frac{1}{f}} + c_3 \log_{\frac{6}{5}}\left(\frac{1}{f}\right) \, ,
    \label{eq:QGrover_simple_ubound}
\end{equation}
where $f = t/|L|$, that upper bounds $E_{\text{Grover}}(|L|,t)$ on the entire domain $[1, |L|]$ of $t$ --- except for the case of $t=0$, which will be dealt with separately in Section~\ref{sec:estimating_marked_items}. Afterwards, we will use Lemmas~\ref{lem:sampling_constant} and~\ref{lem:sampling_constant_log} to deal with the issue of the concavity of the square root and the logarithm appearing in the expression for $G$.

Let us start with the regime $1 \leq t \leq |L|/4$. From Eq.~\eqref{eq:QGrover_ubound_methodology}, using Eq.~\eqref{eq:ubound_mt} with $f = \frac{t}{|L|}$, we have
\begin{align*}
    E_{\text{Grover}} (|L|,t)
    &\leq \left(\frac{3}{2} \sqrt{\frac{3}{f}}+ \log_{\frac{6}{5}}\left(\sqrt{\frac{1}{3f}} \right)  - 2 \right) \left(1 +\frac{1}{1 - \frac{1}{3\sqrt{f |L|}}} \right) \\
    & \leq \left(\frac{3}{2} \sqrt{\frac{3}{f}}+ \log_{\frac{6}{5}}\left(\sqrt{\frac{1}{f}} \right) +\log_{\frac{6}{5}}\left(\frac{1}{\sqrt{3}} \right)  - 2 \right) \left(2 +  \frac{1}{2\sqrt{f |L|}} \right), \numberthis \label{eq:QGrover_estimate_bound}
\end{align*}
where we have bounded the term $\frac{1}{1-\frac{1}{3\sqrt{f |L|}}}$ by $1+ \frac{1}{2\sqrt{f |L|}}$, which holds since (defining $x := \frac{1}{3\sqrt{f |L|}}$)
\begin{equation}
    \frac{1}{1-\frac{1}{3\sqrt{f |L|}}} = \frac{1}{1-x} = 1 + \frac{x}{1-x} \leq 1 + \frac{3}{2}x = 1 + \frac{1}{2\sqrt{f |L|}}, 
    \label{eq:1overx_series}
\end{equation}
when $0 \leq x \leq 1/3$. Plugging Eq.~\eqref{eq:1overx_series} into Eq.~\eqref{eq:QGrover_estimate_bound} and expanding all terms we obtain $E_{\text{Grover}}(|L|,t) \leq G(f)$ by choosing the $c_i$'s as follows.\footnote{Note that, despite there being a term linear in $\frac{1}{f}$, there number of iterations scales as $\sqrt{|L|}$: indeed, the constant $c_1$ multiplying the linear term contains a factor of $\frac{1}{\sqrt{|L|}}$, and since $1/f \leq |L|$, we have that $c_1 \frac{1}{f} = O(\sqrt{|L|})$.} 
\begin{align*}
    c_0 &= 2\log_{\frac{6}{5}}\left(\frac{1}{\sqrt{3}}\right) - 4 \leq -10.0256 \\
    c_1 &= \frac{3\sqrt{3}}{4\sqrt{|L|}} \leq \frac{1.2991}{\sqrt{|L|}}\\
    c_2 &= 3\sqrt{3} 
    + \frac{1}{2\sqrt{|L|}}\left(\log_{\frac{6}{5}}\left(\frac{1}{\sqrt{3}}\right) - 2 \right) \leq 5.1962 - \frac{2.5064}{\sqrt{|L|}}\\
    c_3 &= \frac{5}{4} \,.
\end{align*}
Above, we have absorbed the product of the $\log_{\frac{6}{5}}(\sqrt{1/f})$ and the $1/(2\sqrt{f|L|})$ into $c_3$, by upper bounding $1/(2\sqrt{f|L|}) \leq 1/2 $.

The function $G(f)$ with the constants as given above gives an expression that upper bounds $E_{\text{Grover}}$ on the domain $1\leq t < \frac{|L|}{4}$ (or $4 \leq \frac{1}{f} \leq |L|$), but it does not upper bound $E_{\text{Grover}}$ for $\frac{|L|}{4} \leq t \leq |L|$ (or $1\leq \frac{1}{f} \leq 4$). In order to obtain a single expression that upper bounds $E_{\text{Grover}}$ for all $1 \leq t \leq |L|$, or equivalently $1 \leq \frac{1}{f} \leq |L|$, we can take the expression from Eq.~\eqref{eq:QGrover_simple_ubound} and add an unknown constant\footnote{We can also alter the other coefficients, but that will give a much worse upper bound for the small $t$ regime.}  to $c_0$ to ensure that the resulting expression also upper bounds $E_{\text{Grover}}$ for $\frac{|L|}{4} \leq t \leq |L|$, i.e. $1 \leq \frac{1}{f} \leq 4$. The latter is given by Eq.~\eqref{eq:QGrover_ubound_methodology}:
\begin{equation}
    2.0344 \left(1 +  \frac{1}{1 - \frac{2.0344}{\alpha\sqrt{|L|}}}\right) \, .
    \label{eq:constant_small_t}
\end{equation}
To obtain a bound that holds for all $1 \leq \frac{1}{f} \leq |L|$, note that Eq.~\eqref{eq:QGrover_simple_ubound} is monotonically increasing in $1/f$. Therefore, all we need is that Eq.~\eqref{eq:QGrover_simple_ubound} upper bounds the expression in Eq.~\eqref{eq:constant_small_t} for $\frac{1}{f} = 1$, that is, we require
\begin{equation}
    2.0344 \left(1 +  \frac{1}{1 - \frac{2.0344}{\alpha\sqrt{|L|}}}\right) \leq c_0 + c_1 + c_2 \, ,
    \label{eq:constants_requirement}
\end{equation}
with
\begin{align*}
    c_0 &= 2\log_{\frac{6}{5}}\left(\frac{1}{\sqrt{3}}\right) - 4 + A + \frac{B}{2\sqrt{L}} \numberthis \label{eq:c_0_A_B}\\
    c_1 &= \frac{3\sqrt{3}}{4\sqrt{|L|}}\\
    c_2 &= 3\sqrt{3} 
    + \frac{1}{2\sqrt{|L|}}\left(\log_{\frac{6}{5}}\left(\frac{1}{\sqrt{3}}\right) - 2 \right)\\
\end{align*}
where $A$ and $B$ will be chosen to make Eq.~\eqref{eq:constants_requirement} hold.

Using the same reasoning as for Eq.~\eqref{eq:1overx_series} with $x= \frac{2.0344}{\alpha \sqrt{|L|}}$ and $0 \leq x \leq \frac{2.0344}{\alpha}$, we have
\begin{align*}
        2.0344 \left(1 +  \frac{1}{1 - \frac{2.0344}{\alpha\sqrt{|L|}}}\right) 
        &\leq 2.0344 \left(2 + \frac{1}{1 - \frac{2.0344}{\alpha}} \frac{2.0344}{\alpha \sqrt{|L|}} \right)  \\
        &= 4.0688 + \frac{1}{\sqrt{|L|}} \frac{2.0344^2}{\alpha - 2.0344}  
\end{align*}
In order to make the right-hand side of the above expression less than or equal to $c_0 + c_1 + c_2$, we can compare the constant terms and the terms multiplying $\frac{1}{\sqrt{|L}}$ on both sides separately. For the constant terms, we require
\[
    4.0688  \leq 2 \log_{\frac{6}{5}}\left(\frac{1}{\sqrt{3}}\right) - 4 + A + 3\sqrt{3}
\]
which implies that
\[
    A = 8.8984 \geq 4.0688 - 2 \log_{\frac{6}{5}}\left(\frac{1}{\sqrt{3}}\right) + 4 - 3\sqrt{3}
\]
works. For the terms multiplying $\frac{1}{\sqrt{|L|}}$, we need
\[
     \frac{2.0344^2}{\alpha - 2.0344}  
    \leq \frac{1}{2}\left(\frac{3\sqrt{3}}{2} 
    +  \log_{\frac{6}{5}}\left(\frac{1}{\sqrt{3}}\right) - 2 + B  \right) \,
\]
which means we can take
\[
    B = 3.5700  \geq 2\, \frac{2.0344^2}{\alpha - 2.0344}  
    - \left(\frac{3\sqrt{3}}{2} + \log_{\frac{6}{5}}\left(\frac{1}{\sqrt{3}}\right) - 2 \right).
\]

Taking everything together, we obtain the following upper bound
\begin{equation*}
    E_{\text{Grover}}(|L|, t) \leq -1.1272 + \frac{1.7850}{\sqrt{|L|}} + \frac{1.2991}{\sqrt{|L|}} \frac{1}{f} + \left(5.1962 - \frac{2.5064}{\sqrt{|L|}}\right) \sqrt{\frac{1}{f}} +  \frac{5}{4} \log_{\frac{6}{5}}\left( \frac{1}{f}\right)\, .
\end{equation*}
Hence, the following estimator
\begin{equation*}
    \tilde{E}_{\text{Grover}}^{\text{estimator}}(l) := -1.1272 + \frac{1.7850}{\sqrt{|L|}} + \frac{1.2991}{\sqrt{|L|}} l + \left(5.1962 - \frac{2.5064}{\sqrt{|L|}}\right) \sqrt{l} +  \frac{5}{4} \log_{\frac{6}{5}}\left( l\right)\, 
\end{equation*}
satisfies
\[
    \tilde{E}_{\text{Grover}}^{\text{estimator}}(\E[l]) = \tilde{E}_{\text{Grover}}^{\text{estimator}}(1/f) \geq E_{\text{Grover}}(|L|,t) \, .
\]
for all $1 \leq t \leq |L|$.

% Plugging the obtained expressions for $A$ and $B$ into Eq.~\eqref{eq:c_0_A_B} yields the following estimator:
% \[
%     \tilde{E}_{\text{Grover}}^{\text{estimator}}(l) :=-1.1272 + \frac{1.7850}{\sqrt{|L|}} + \frac{1.2991}{\sqrt{|L|}} l + \left(5.1962 - \frac{2.5064}{\sqrt{|L|}}\right) \sqrt{l} +  \frac{5}{4} \log_{\frac{6}{5}}\left( l\right)
% \]

\

\noindent Next, we want to construct an estimator $E_{\text{Grover}}^{\text{estimator}}(l)$ that satisfies
\[
    \E[E_{\text{Grover}}^{\text{estimator}}(l)] \geq \tilde{E}_{\text{Grover}}^{\text{estimator}}(\E[l]) \, .
\]
Lemma~\ref{lem:sampling_constant} implies that, if we  multiply our sample average\footnote{Of sample size 1.} of $\sqrt{l}$ by $2/\sqrt{\pi}=\sqrt{d_1}$, we obtain an an estimator that in expectation upper bounds $\sqrt{1/f}$. Similarly, from Lemma~\ref{lem:sampling_constant_log} we gather that, if we multiply $l$ by $e^\gamma$ in the argument of the logarithm, we obtain an estimator that in expectation upper bounds $\log_{\frac{6}{5}}\left(\frac{1}{f}\right)$. Consequently, the following estimator 
\begin{equation*}
    E_{\text{Grover}}^{\text{estimator}}(l) :=
    -1.1272 + \frac{1.7850}{\sqrt{|L|}} + \frac{1.2991}{\sqrt{|L|}} l + \left(5.1962 - \frac{2.5064}{\sqrt{|L|}}\right) \frac{2\sqrt{l}}{\sqrt{\pi}} +  \frac{5}{4} \log_{\frac{6}{5}}\left(e^\gamma l\right)
\end{equation*}
upper bounds $E_{\text{Grover}}$ in expectation for all $1 \leq t \leq |L|$ (or $1 \leq \frac{1}{f} \leq |L|$):
\[
    \E[E_{\text{Grover}}^{\text{estimator}}(l)] \geq E_{\text{Grover}}(|L|,t) \, ,
\]
where the expectation is taken over the geometric distribution $l \sim \text{Geo}(f)$, with $f = t / |L|$.

\subsection{Estimator for expected number of queries of \textbf{QSearch}}
\label{app:Estimator-QSearch}

In this section we will prove that, for a list $L$ with $t \geq 1$ marked items, 
\begin{equation}
    \E[H(l)] \geq \frac{|L|}{t}\left(1 - \left(1-\frac{t}{|L|}\right)^{N}\right) + \left(1-\frac{t}{|L|}\right)^{N} c_q E_{\text{Grover}}(|L|,t) \, ,
    \label{eq:exp_Hl_to_show}
\end{equation}
where $H(l) = h_1(l) + h_2(l)c_q E_{\text{Grover}}^{\text{estimator}}(l)$, $N = N_{\text{samples}}$ is the number of classical samples taken by \textbf{QSearch}, and the expectation over $l$ is taken over the geometric distribution $l \sim \text{Geo}(f)$, with $f = t/|L|$.

As a quick sanity check, if $l$ is geometrically distributed with parameter $f$ (e.g. $\Pr[x=k] = f(1-f)^{k-1}$), then we have
\[
    \E[l] = \sum_{k=1}^{\infty} (1-f)^{k-1} f k = - f \frac{\partial}{\partial f} \sum_{k=0}^{\infty} (1-f)^k = -f \frac{\partial}{\partial f} \frac{1}{f} = \frac{1}{f}\, ,
\]
as expected.

Next, we focus on the first term in Eq.~\eqref{eq:exp_Hl_to_show}. Recall from Eq.~\eqref{eq:Ef_series} that for the classical contribution to the queries, we have 
\begin{equation}\label{eq:purple_tortoise}
    \frac{1}{f}\left( 1 - (1-f)^N\right) = \sum_{i=1}^N f(1-f)^{i-1}i + N(1-f)^N\,.
\end{equation}
Now, for $h_1(l) = \min(l,N)$, we have the following lemma.
\begin{lemma}
    $\E[h_1(l)] = \frac{1}{f}\left( 1 - (1-f)^N\right)$. (Where $x$ is drawn according to the geometric distribution above with parameter $f$).
\end{lemma}
\begin{proof}
    \begin{eqnarray*}
        \E[g(l)] &=& \sum_{k=1}^{\infty} (1-f)^{k-1} f h_1(k) = \sum_{k=1}^N (1-f)^{k-1}fk + \sum_{k=N+1}^{\infty} (1-f)^{k-1}fN \\
        &=& \sum_{k=1}^N (1-f)^k f k + N(1-f)^N \sum_{k=1}^{\infty}(1-f)^{k-1} f \\
        &=& \sum_{k=1}^N (1-f)^k f k + N(1-f)^N \\
        &=& \frac{1}{f}\left( 1 - (1-f)^N\right) \,,
    \end{eqnarray*}
where the final equality holds from \eqref{eq:purple_tortoise}.
\end{proof}
Hence, we can take as an estimator for the number of classical queries just $\min(l,N)$, where $l$ is the number of items sampled from $L$ before finding a marked one. 

Now we turn our attention to the quantum contribution to the number of queries, which is given by 
\[
    (1-f)^N c_q E_{\text{Grover}}(|L|, t)\,.
\]
We already have from Eq.~\eqref{eq:QGrover_estimator_final} an estimator $E_{\text{Grover}}^{\text{estimator}}(l)$ such that $\E\left[E_{\text{Grover}}^{\text{estimator}}(l) \right] \geq  E_{\text{Grover}}(|L|, t)$. We seek a function $h_2$ such that, when multiplied by $E_{\text{Grover}}^{\text{estimator}}$ we also have $\E\left[h_2(l)E_{\text{Grover}}^{\text{estimator}}(l)\right] \geq (1-f)^N E_{\text{Grover}}(|L|,t)$. Toward this end, let 
\[
    h_2(l) = \begin{cases}
        0 \qquad l\leq N \\
        1 \qquad l>N.
    \end{cases}
\]
Then 
\begin{eqnarray*}
    \E[h_2(l)] &=& \sum_{k=1}^{\infty} (1-f)^{k-1} f h_2(k) = \sum_{k=N+1}^{\infty} (1-f)^{k-1}f \\
    &=& (1-f)^N \sum_{k=1}^{\infty} (1-f)^{k-1} f = (1-f)^N\,.
\end{eqnarray*}
It remains to show that $\E\left[h_2(l)E_{\text{Grover}}^{\text{estimator}}(l)\right] \geq (1-f)^N E_{\text{Grover}}(|L|,t)$. 
\begin{lemma}\label{lem:multiplicative_expectations}
    Let $f,g:\bbN \rightarrow \R$ be non-negative non-decreasing functions, and $x$ a random variable on $\bbN$. Then
    \[
        \E[f(x)g(x)] \geq \E[f(x)]\E[g(x)]\,.
    \]
\end{lemma}
\begin{proof}
Fix $x,y, \in \bbN$. Because $f$ and $g$ are non-decreasing, we have
\[
    (f(x) - f(y))(g(x) - g(y)) \geq 0 \, ,
\]
and therefore
\[
    f(x)g(x) + f(y)g(y) \geq f(x)g(y) + f(y)g(x)\, .
\]
Because $f$ and $g$ are non-negative functions on $\bbN$, $f$, $g$ and the product $fg$ are Lesbesgue integrable with respect to the probability measure. Taking the expectation value over $x$ and $y$ yields
\[
    \E_x[f(x)g(x)] + \E_y[f(y)g(y)] \geq \E_x[g(x)]\E_y[g(y)] + \E_y[f(y)]\E_x[g(x)]\, ,
\]
from which the lemma follows.
\end{proof}
Combining Lemma~\ref{lem:multiplicative_expectations} with the observations above, as well as the fact that both $h_2$ and $E_{\text{Grover}}^{\text{estimator}}$ are non-decreasing, we conclude that the function $H(l) = h_1(l) + h_2(l)c_q E_{\text{Grover}}^{\text{estimator}}(l)$ satisfies
\[
    \E[H(l)] \geq \frac{1}{f}\left(1-(1-f)^N\right) + (1-f)^N c_q E_{\text{Grover}}(|L|,t) \,.
\]
for $t \geq 1$. Hence, $H$ is an estimator that always upper bounds, in expectation, the number of queries made by QSearch when there is at least one marked item.

\end{document}